\documentclass[12pt,onecolumn,journal]{IEEEtran}
\usepackage{citesort}
\usepackage[dvips]{graphicx}
\usepackage[cmex10]{amsmath}
\usepackage{algorithm}
\usepackage{algpseudocode}
\usepackage{array}
\usepackage[tight,footnotesize]{subfigure}
\usepackage{setspace}
\usepackage{enumerate} 
\usepackage{url}
\usepackage{empheq}
\usepackage{amsfonts, amssymb, amsthm, mathrsfs}
\usepackage{setspace}
\usepackage{Tabbing}
\usepackage{fancyhdr}
\usepackage{lastpage}
\usepackage{extramarks}
\usepackage{chngpage}
\usepackage{color}
\usepackage{indentfirst}
\usepackage{float,wrapfig}
\usepackage{times}
\usepackage{multirow}
\usepackage{bm}
\usepackage{bbm}
\usepackage{mathtools}
\usepackage{slashbox}
\usepackage{lipsum}

\DeclareMathOperator*{\argmax}{arg\,max}

\makeatletter
\g@addto@macro{\normalsize}{%
   \setlength{\abovedisplayskip}{3pt plus 2pt minus 2pt}
   \setlength{\abovedisplayshortskip}{3pt plus 2pt minus 2pt}
   \setlength{\belowdisplayskip}{3pt plus 2pt minus 2pt}
   \setlength{\belowdisplayshortskip}{3pt plus 2pt minus 2pt}
   \setlength{\textfloatsep}{10pt plus 2pt minus 2pt}
   }
\makeatother

\newtheorem{theo}{Theorem}
\newtheorem{defi}{Definition}

\newtheorem{prop}{Proposition}
\newtheorem{remark}{Remark}

\newcommand{\mb}{\mathbf}        
\newcommand{\conv}{\textrm{conv}} 
\newcommand{\Csize}{L}
\newcommand{\Uc}{\mathcal{U}}   
\newcommand{\Bc}{\mathcal{B}}             
\newcommand{\Cc}{\mathcal{C}}

\newcommand{\Kc}{\mathcal{K}}
\begin{document}

\title{User Association and Interference Management in Massive MIMO HetNets }

\author{Qiaoyang Ye, 
			Ozgun Y. Bursalioglu, 
			Haralabos C. Papadopoulos, 
			Constantine Caramanis
        	and Jeffrey G. Andrews
\thanks{Q. Ye, C. Caramanis and  J. G. Andrews are with WNCG, The University of Texas at Austin, Austin, TX, USA, O. Y. Bursalioglu and H. C. Papadopoulos are with Docomo Innovations Inc, Palo Alto, CA, USA.  Email: 
qye@utexas.edu,  \{obursalioglu, hpapadopoulos\}@docomoinnovations.com, constantine@utexas.edu, jandrews@ece.utexas.edu. 
Manuscript last revised: \today.}}

\maketitle

\begin{abstract}
Two key traits of 5G cellular networks are much higher base station (BS) densities -- especially in the case of low-power BSs -- and the use of massive MIMO at these BSs. This paper explores how massive MIMO can be used to jointly maximize the offloading gains and minimize the interference challenges arising from adding small cells. We consider two interference management approaches: joint transmission (JT) with local precoding, where users are served simultaneously by multiple BSs without requiring channel state information exchanges among cooperating BSs, and resource blanking, where some macro BS resources are left blank to reduce the interference in the small cell downlink. A key advantage offered by massive MIMO is channel hardening, which enables to predict instantaneous rates a priori. This allows us to develop a unified framework, where resource allocation is cast as a network utility maximization (NUM) problem, and to demonstrate large gains in cell-edge rates based on the NUM solution. We propose an efficient dual subgradient based algorithm, which converges towards the NUM solution. A scheduling scheme is also proposed to approach the NUM solution. Simulations illustrate more than 2x rate gain for 10th percentile users vs. an optimal association without interference management.
\end{abstract}


\section{Introduction}\label{sec:introduction}
Smart densification \& heterogeneity in  base station (BS) deployments and massive MIMO are considered as two of the most important technologies in 5G cellular systems~\cite{And7ways13,And14,BocEtAl14}.\footnote{As higher-frequency spectrum being available, large arrays become practical even for small cells. For example, at 3.5GHz band, a 36-antenna (arranged on a square grid at half-wavelength separation) can be implemented on a 26cm $\times$ 26cm surface. The required implementation area would be  smaller  as the carrier frequency becomes higher.} 
The massive MIMO regime is the setting when the number of antennas at a BS is significantly larger than the number of users that are simultaneously served by the BS \cite{Mar10,HoyTen13,LarEdf14}. In this paper, various aspects such as user association, load balancing, scheduling and interference management are considered for future networks with massive MIMO deployments where these technologies are adapted together.

\vspace{-0.4cm}
\subsection{Motivation and Related Work}\label{sec:relatedwork}
Conventionally, mobile user equipments (UEs) are served by the BS providing the largest signal-to-interference-plus-noise ratio (SINR) or the largest received power \cite{3GPP872} -- called \emph{max-SINR association} in this paper. 
In heterogeneous networks (HetNets), however, different types of BSs can have large disparities in transmit power, so a max-SINR association results in heavily congested macrocell BSs and lightly loaded low-power small cell BSs.  This results in a  very inefficient use of available time-frequency resources, and strongly motivates load balancing, which in effect means pushing some UE traffic onto lightly loaded small cells even if it requires reducing their SINRs by many dBs~\cite{AndLBoverview14}.  

\textbf{Load Balancing.} Several approaches have been used to study load balancing in HetNets, including stochastic geometry \cite{SinDhi13,JoSan12}, game theory~\cite{AryKes13} and system-level simulations \cite{DamMon11,GhoAnd12}.   Meanwhile, in industry, proactive load balancing is accomplished by \emph{biasing} UE association towards the small cells \cite{DamMon11,GhoAnd12}.\footnote{Biasing refers to artificially adding a bias value (e.g., $10$ dB) to received signal power from small cell layer at UEs. } Our initial study on load balancing \cite{YeAnd13LB} formulated a network utility maximization (NUM) problem for user association in HetNets with single-antenna BSs, where resources are equally allocated among users in the same cell.  
The equal resource allocation can be suboptimal if the user associations happen on a much slower time scale than the channel variations. 
In general, the user association and scheduling (resource allocation) problems are coupled, and it is quite difficult to jointly optimize them.

\textbf{Massive MIMO}.  A key benefit of massive MIMO is that the extra diversity afforded by the large antenna array averages out the fast fading, and thus the instantaneous rate stabilizes to the long-term mean which changes on slow time scales.  
As shown in \cite{BetBur14a}, the instantaneous rates can be predicted with peak-rate proxies, which are independent of scheduled instances. 
This property allows the decoupling of user association and scheduling, which is exploited to achieve near-optimal load balancing in massive MIMO HetNets with cellular transmission (where data for each user is transmitted from a single BS)~\cite{BetBur14a}.

\textbf{Coordinated multi-point transmission}. MIMO techniques also provide the option of serving a user at high rates from multiple BSs -- referred to as coordinated multi-point transmission (CoMP), which is proposed as one of the core features in LTE-Advanced \cite{GesHan10,SawKis10,LeeSeo12}. The set of BSs that cooperatively serve the same user is referred in this paper as a \emph{BS cluster}. Paper \cite{MarFet11} studies how to determine the BS clusters, while~\cite{LiSve11,ZhaQue13,DudeV14,BjoKou13,HonSun13,SanRaz14,LiBjo15} investigate jointly optimized designs involving some of the following aspects: BS cluster selection, beamforming (e.g., coordinated beamforming and joint transmissions), user scheduling and power allocation. In studies such as \cite{BjoKou13,HonSun13,SanRaz14,LiBjo15}, the complexity of proposed algorithms can become prohibitive as number of antennas increases. On the other hand, an efficient suboptimal precoder for the single-cell scenario with large antenna arrays is proposed in \cite{BjoKou13}.


\textbf{Resource Blanking}.  As the macrocells that users are offloaded from now become strong interferers for these offloaded users, the increased interference eats into the gains offered by load balancing. This motivates us to jointly consider user association and interference management.  Besides CoMP, another popular interference management approach is to leave some macro resource blocks (RBs) blank, similar to enhanced intercell interference coordination (eICIC) in 3GPP \cite{LopGuv11}. The key difference between RB blanking in our work and eICIC is that eICIC focuses on the time domain, while in this work blanking is applied in both time and frequency domains. We call the RBs where macro BSs are muted the \emph{blank RBs}, while the rest of RBs are called \emph{normal RBs}.  Several works have considered the joint problem of user association and RB blanking. For example, \cite{VasPup13,LiuLau14} proposes a dynamic approach adapting the muting duty cycle to load variations, while \cite{YeAnd13ABS,BedAgr13,GhiRos13,DebMon14,SinAnd14} consider a more static approach.

For general multi-cell massive MIMO HetNets, the joint optimization of user association and interference management including joint transmission and resource blanking is still an open issue. In this work, we combine various aspects of resource allocation and interference management including resource blanking, joint transmission, association, user scheduling, etc. for massive MIMO deployments.  
We focus on a  distributed-MIMO form of CoMP, which allows local precoding at each BS and does not require channel state information (CSI) exchanges among cooperating BSs \cite{BjoZak10}. We call this specific form of CoMP as Local Joint Transmission (LJT). LJT allows us to  develop a systematic resource allocation approach for CoMP (including cellular transmission as a special case). 
Other interference management approaches (e.g. \cite{HosRas14}) can also be adopted at the cost of additional complexity and overheads (e.g., schemes with joint precoding as discussed in \cite{LeeSeo12}), but such designs are beyond the scope of this paper.

\textbf{Cross-layer Optimization.}  
To study the joint user association and interference management problem, we propose to use the  cross-layer optimization approach, aiming to improve the rate distribution, particularly the cell-edge performance. 
Cross-layer optimization is quite popular in the study of resource allocation (see, e.g., \cite{ShaRap03,LinShr06} and references therein). 
Among these, besides studies with disjoint clusters (i.e., each BS belongs to at most  one cluster on any RB), e.g.  \cite{DudeV14}, user-specific clusters with overlapping BSs have also been considered \cite{Bjornson13,BjoJal11}. At any scheduling instant, the cluster formation method we consider can be described using ``Dynamic Cooperation Clusters (DCC)'' concept of \cite{Bjornson13}. An important difference between works following the DCC scheme \cite{Bjornson13} and our work resides in the selection of which BSs serve which users (which also inherently specify active clusters at an RB). For the former case, this selection is based on instantaneous channel gains between users and BSs as in \cite{BjoJal11}. In particular, the resource allocation in  \cite{BjoJal11,Bjornson13} (consisting of as precoding design, scheduling and power allocation) is done at each RB to optimize, for example, instantaneous user rates for some utility function. 
On the other hand,  in our setting massive MIMO rate hardening allows us to allocate resources over many RBs to optimize a utility function over long-term user rates. Thus, the BS cluster selection for each user is determined as a result of load balancing and resource allocation across many RBs,  which are performed ahead of time at a much coarser time scale than an individual-RB time scale.
We then present scheduling policies at a finer time scale (i.e., RB level) to approach the optimized (coarser time scale) resource allocation.\footnote{ In this paper similar to \cite{LiuLau14}, we also consider resource allocation over two different time scales. \cite{LiuLau14} exploits the sparsity of the interference graph of the HetNet topology to overcome the complex coupling between user scheduling and RB blanking. }

\vspace{-0.4cm}
\subsection{Contributions and Organization}
In this paper, we present a novel framework for the joint optimization of user association and interference management in massive MIMO HetNets, resulting in the following main contributions.


\textbf{A unified NUM problem.} 
By exploiting the predictable instantaneous rate, user association and scheduling problems can be decoupled, allowing us in Sec. \ref{sec:formulation} to formulate a unified convex optimization problem for resource allocation with both LJT and RB blanking. Note that in the considered LJT, the clusters are user-specific (i.e., different users can be served by different clusters).
The formulated NUM problem can also be applied to scenarios where some bandwidth resources are explicitly reserved for macro or small cell operation, while some resources are reserved for being shared by both layers. As an extension of \cite{BetBur14a}, the optimal solutions can always be realized by a suitably designed scheduler when blanking is used in cellular transmission. On the other hand, with LJT, we show  that there exist some solutions that are not implementable. Naturally, the solution of the NUM problem -- called the \emph{NUM solution} -- upper bounds the network performance and can serve as a useful benchmark. 

\textbf{Dual subgradient based algorithm.}  
Sec. \ref{sec:dual-algo} presents an efficient algorithm based on the dual sub-gradient method, which converges towards the optimal dual variables. As the objective function is not strictly convex, it is difficult to get the optimal primal variables given optimal dual variables. Exploring the solution structure, we formulate a small-size linear program (LP) to get the optimal primal variables. 


\textbf{Simple scheduling scheme to approach the NUM solution.} 
 Note that the NUM solution provides the desirable resource allocation at a coarser time scale. In Sec. VI, we further present scheduling policies at a finer time scale (i.e. RB level), which target approaching the optimized resource allocation in the long term.


Simulations\footnote{Some of these results are also published in \cite{YeBursaPapa15}.} in Sec. \ref{sec:simulation} show that the proposed  harmonized CoMP/cellular operation can  provide significant gains with respect to cellular-only massive MIMO operation \cite{BetBur14a}, especially for cell-edge users.
For example, the rate of  bottom (10th percentile)  users  in our setup is about 2.2$\times$ with respect to the optimal user association without interference management, which itself is much larger than the max-SINR association. Also, the utility provided by the proposed  scheduling scheme is within 90\% of the utility provided by the NUM~solution.

For convenience, the key notation in this paper is summarized in Table \ref{tb:notation}.
\begin{table}
\label{tb:notation}
\caption{Notation Summary}
\begin{center}
\begin{tabular}{|c||c|}
\hline
Notation & Description\\
\hline
$\mathcal{B}_{\rm m} $, $\mathcal{B}_{\rm s}$, $\mathcal{B} $, $\mathcal{U}$ & set of macro BSs, small cell BSs, all BSs, UEs, respectively\\
\hline
$M_j $ & number of antennas at BS $j$\\
\hline
$P_j$ & transmit power of BS $j$\\
\hline
 $\mathbf{G}_j$ & channel matrix between BS $j$ and its users\\
 \hline
 $\mathbf{g}_{kj}$ & channel vector between BS $j$ and UE $k$\\
 \hline
 $g_{kj,i},h_{kj,i}$ & channel, fast fading between $i$th antenna of BS $j$ and UE $k$, respectively \\
\hline
$ \beta_{kj}$ & slow fading between BS $j$ and UE $k$\\
\hline
$\mathbf{f}_{kj}$ & precoder of BS $j$ for UE $k$\\
\hline
$w_k, \sigma^2$ & AWGN, variance of AWGN, respectively\\
\hline
$L$ & cluster size\\
\hline
$S_j$ & number of users that can be simultaneously served by BS $j$ in cellular transmission\\
 \hline
$S_j(L)$ & number of users that can be simultaneously served by BS $j$ in cluster of size $L$\\
\hline
$A$& operation option\\
\hline
$\mathcal{B}^{(A)}$& set of active BSs in operation $A$\\
\hline
$L_{\max}^{(A)}$ & maximal cluster size in operation $A$\\
\hline
$\Cc$ & cluster index\\
\hline
$\mathcal{U}_\Cc$ & set of users served by cluster $\Cc$ \\
\hline
$y_k(t)$ & received signal at UE $k$ on RB $t$\\
\hline
$s_k$ & transmitted signal for UE $k$\\
\hline
$r_{k\Cc}^{(A)}$ &  instantaneous rate of user $k$ from cluster $\Cc$ in Band-$A$ in massive MIMO regime\\
\hline
$x_{k\Cc}^{(A)}$ & activity fraction of UE $k$ from cluster $\Cc$ in Band-$A$\\
\hline
$R_k$ & long-term rate of UE $k$\\
\hline
$\mu_A$ & fraction of RBs allocated to Band-$A$\\   
\hline
$\lambda_{AL}$ & fraction of RBs allocated to size-$L$ clusters in Band-$A$\\
\hline
$\tilde{x}_{k\Cc}^{(A)}$ & approximate activity fraction by unique association\\
\hline
$\Cc^*(k)$ & the cluster that serves UE $k$ in the unique association\\
\hline
$\alpha_k$ & the desired fraction of resources for UE $k$ in the considered band\\
\hline
$\tilde{R}_k$ & assumed rate for the VQ scheduling, where $\tilde{R}_k = 1/\alpha_k$\\
\hline
$A_{\max}, V$ & sufficiently large parameters in VQ scheduling scheme\\
\hline 
$Q_k(t), A_k(t)$ & VQ length, arrival rate of UE $k$ in VQ scheduling scheme, respectively\\
\hline
$\rho$ & a tunable parameter to characterize how many users the BSs can schedule simultaneously\\
\hline
\end{tabular}
\end{center}
\end{table}

\section{System Model}\label{sec:model}

In this paper, we focus on best-effort traffic and consider downlink (DL) transmission in a HetNet with $J$ BSs and $K$ single-antenna users. We let $j\!\in\!\mathcal{B}\!=\!\{1, 2, \dots, J\}$ and $k\!\in\!\mathcal{U}\!=\!\{1, 2, \dots, K\}$ denoted the indices of BSs and users, respectively. Without loss of generality, we focus on two-tier HetNets comprising a macro layer and a small-cell layer, such as the one considered in Fig. \ref{fig:system-model}. Letting $\mathcal{B}_{\rm m}$ and $\mathcal{B}_{\rm s}$ be the set of macro and small cell BSs, respectively, the BSs belonging to  $\mathcal{B}_{\rm m}$ and $\mathcal{B}_{\rm s}$ can differ in terms of transmit power, size, density, and number of antennas \cite{And7ways13}. 
The number of antennas at BS $j$  is denoted by $M_j$ with $M_j\gg 1$. 
We assume time division duplex (TDD) operation with reciprocity-based CSI acquisition \cite{Mar10,HuhCai11}. Hence, each user sends a single uplink (UL) pilot to train \emph{multiple nearby} BSs. In contrast to feedback-based CSI acquisition, this enables the training of large antenna arrays with overhead proportional to the number of simultaneously served users.  Moreover, it  enables CoMP with practical CSI acquisition overheads.
We also assume a block-fading channel model where the channel coefficients remain constant within each RB \cite{CaiJin10,Mar10,HuhCai11,HoyTen13}.

\begin{figure*}[tp]
\centering
\setcounter{subfigure}{0}
\subfigure[Cellular]{\label{fig:system-model-a}\includegraphics[width=4.52cm, height=3.38cm]{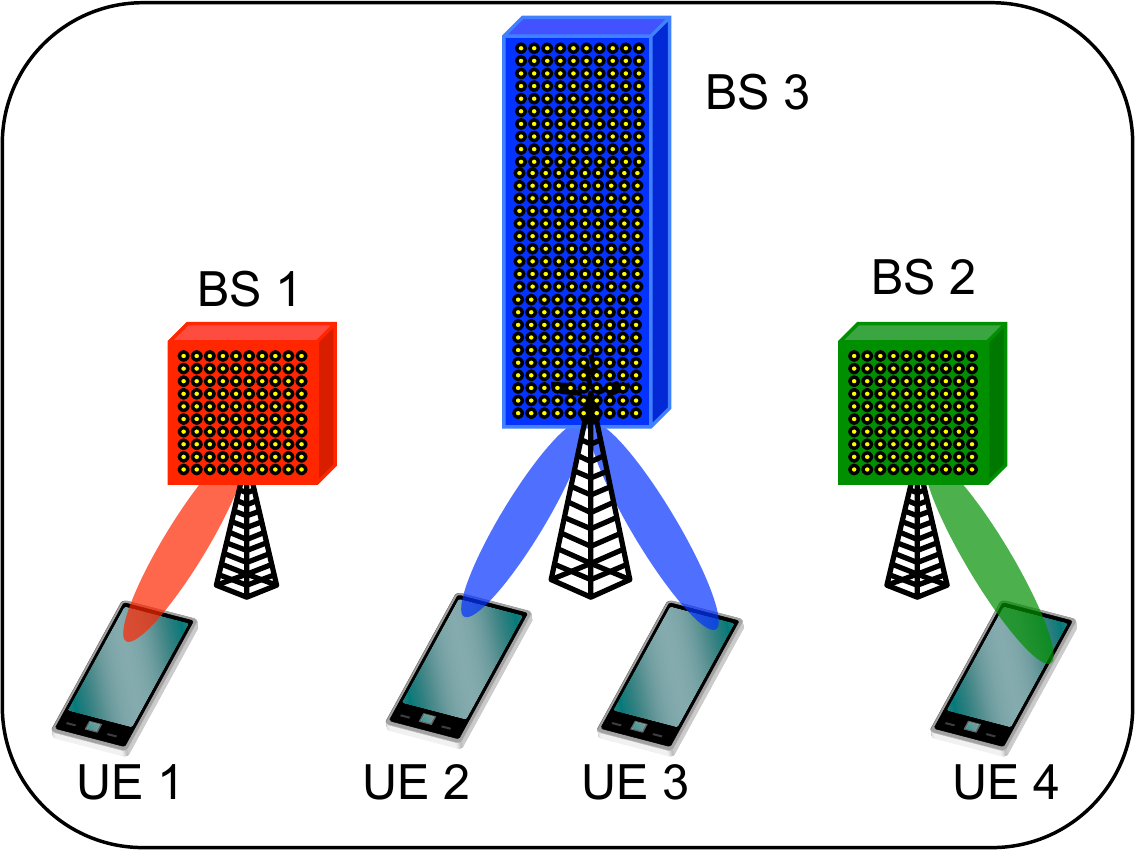}}
 \hspace{0.4in}
\subfigure[Local Joint Transmission]{\label{fig:system-model-b}\includegraphics[width=4.52cm, height=3.38cm]{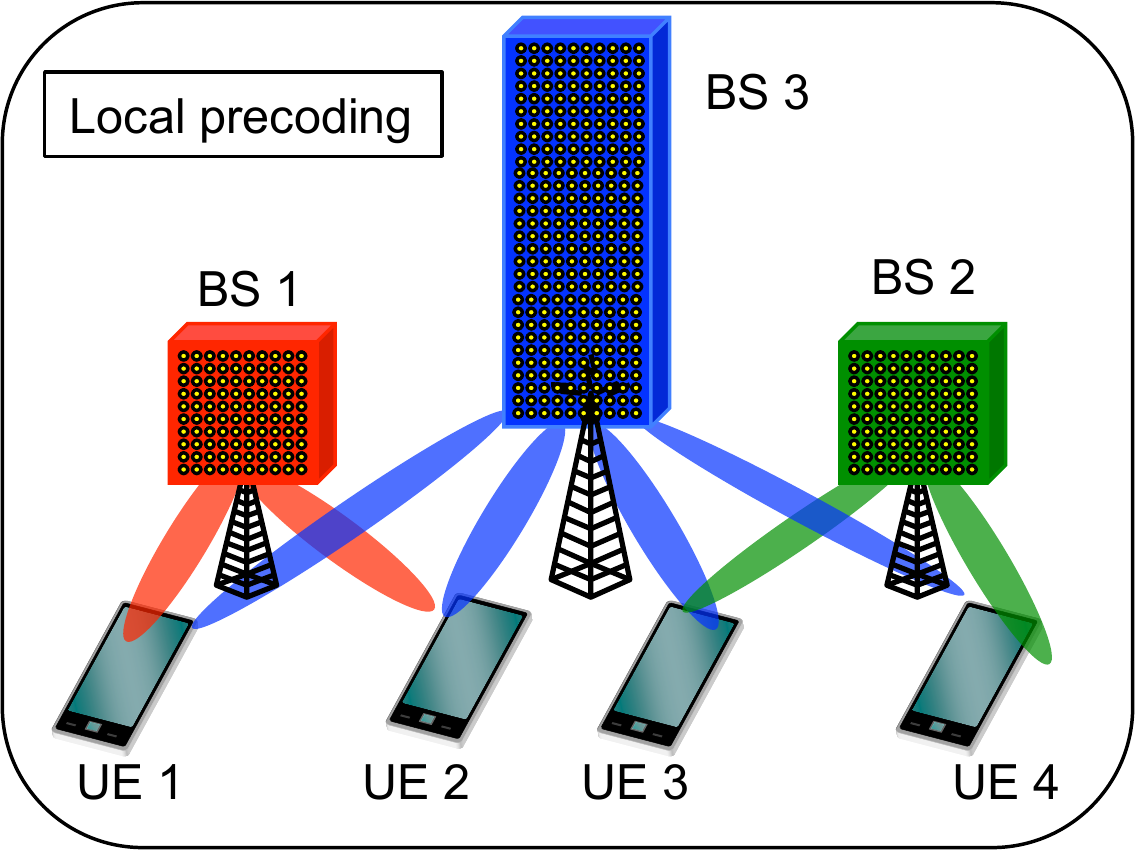}}
 \hspace{0.4in}
\subfigure[Global Joint Transmission]{\label{fig:system-model-c}\includegraphics[width=4.52cm, height=3.38cm]{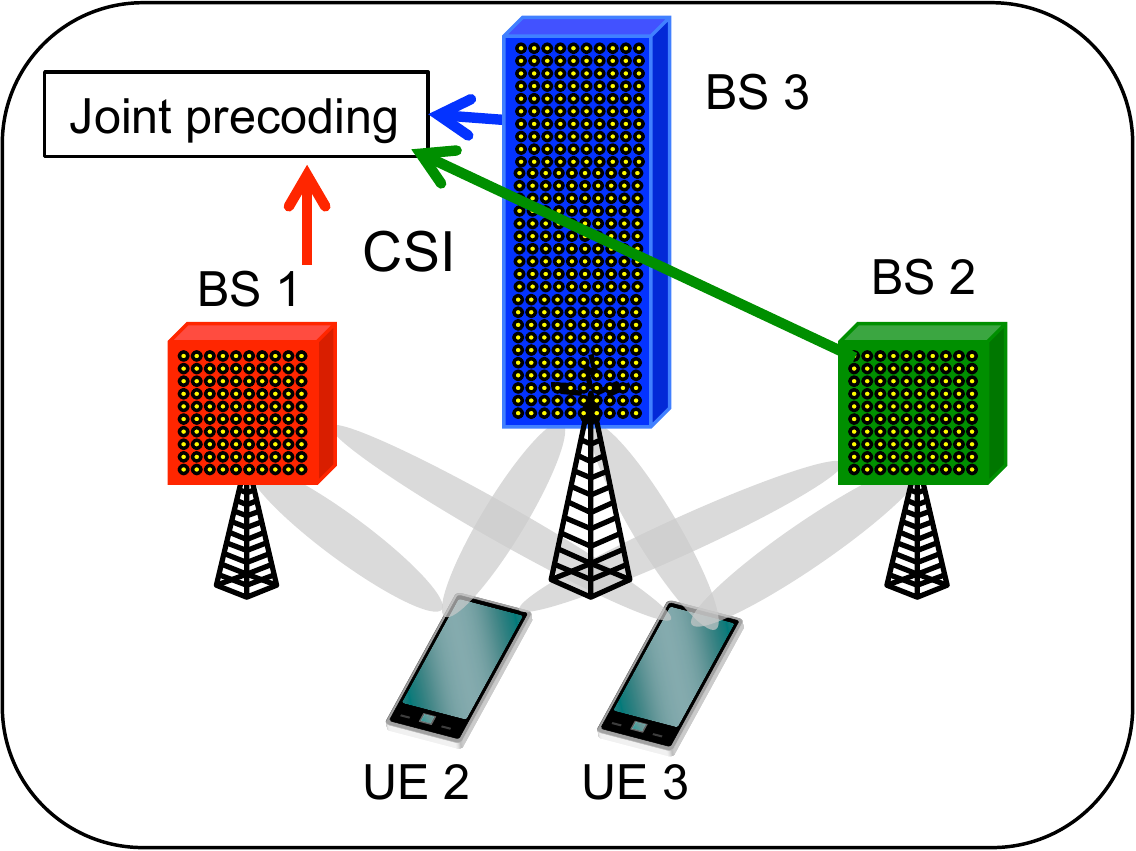}}
\setlength{\abovecaptionskip}{0pt}
\caption{Various MIMO transmission schemes. Color of the beam serving a user indicates where the CSI used for precoding is obtained. For both cellular and local joint transmission, beams are the same color with the BS they are emitted from because in these cases precoding is done only using locally obtained CSI. On the other hand with global joint transmission beams are precoded jointly using CSI from all BSs hence the beams are gray colored.}
\label{fig:system-model}
\end{figure*}

With massive MIMO, a subset of users are scheduled (or, in other words, are active) on each RB. Various transmission schemes in terms of BS-cluster options are possible with this setup as shown in Fig. \ref{fig:system-model}. At one extreme is \emph{cellular transmission} shown in Fig. \ref{fig:system-model-a}, where the coded data for any given active user is transmitted from a single BS and each BS manages interference only to the users it serves. This corresponds to having $|\Bc|$ many disjoint clusters of size $1$. Fig. \ref{fig:system-model-c} shows the other extreme, Global Joint Transmission (ideal/full CoMP) where all BSs serve and coordinate interference to all active users in the network. In this case there is a single cluster serving all of the users. Other transmission schemes with cluster options in between these two extremes are also possible. DCC \cite{Bjornson13} can be used to describe any of these transmission schemes. Fig. \ref{fig:system-model-b} shows the scheme considered in this work where BS-clusters are user-specific. According to this scheme, different active users can be served simultaneously by different potentially overlapping clusters.

CoMP has some well-known advantageous over cellular transmission:\\
 \noindent (1) \textbf{Performance gain at the cell edge:} 
The beamforming (BF) gain becomes intra-cluster BF gain in CoMP, as the same coded data is transmitted from all BSs serving the user. Similarly, the intra-cell interference mitigation is extended across the cluster of BSs by which the user is served.  As a result, the performance gain can be realized at the cell edge.

\noindent (2) \textbf{Low training overhead:}  A single UL pilot  from a user  can be received at all nearby BS antennas, whether these are in the same or different locations. Thus, the CSI acquisition  between a user and nearby BSs need not  incur additional overheads with respect to cellular transmission in the TDD system. The total number of users that can served by the system depends on the number of available UL pilots, and it should be the same regardless of cellular or CoMP if the UL pilots in cellular transmission and CoMP are the same.

A major disadvantage of the general CoMP \cite{ZhaChe09,HuhTul12,LeeSeo12} is that it incurs additional overhead compared to cellular transmission for CSI exchanges among cooperating BSs. To overcome this challenge, we focus on schemes that perform local precoding at each BS: The user beams (i.e., precoding vectors) at any RB for any BS $j$ are designed as if BS $j$ was in cellular multi-user (MU)-MIMO transmission over all the users it serves. All BSs serving a user transmit the same coded data stream to the user \cite{BjoZak10}. Each BS transmits the stream on a beam that is (independently) designed for the users at that BS.
For instance, the beam with Linear Zero forcing Beamforming (LZFBF) for each user served by BS $j$  is chosen within the null space of the channels of all the other users served by BS $j$, no matter whether there are additional BSs serving the user on the same RB or not. Due to local precoding,  BS $j$ only needs CSI between the users it serves and its own antennas to generate the user beams. 
Hence, the challenge of costly CSI exchanges between BSs is eliminated.

In contrast to full CoMP\footnote{With general CoMP schemes (e.g., cluster ZFBF transmission), the instantaneous rate that is provided to a user by a cluster of BSs  is  {\em also} a function of the identity (and in fact the large-scale channel coefficients) of the set of the {\em other}  users scheduled for cluster transmission with the user, and can thus vary from slot to slot depending on the scheduling set \cite{HuhTul12}.  As a result, in contrast to the LJT schemes we consider (whereby a user's instantaneous rate is independent of the identity of the active users in the slot),  general CoMP schemes are not amenable to the load balancing and scheduling techniques presented in this work.}, with local precoding of the form depicted in Fig. \ref{fig:system-model-b}, the  instantaneous rate is \emph{independent}  of the other active users' channels and only depends on the power allocation for user streams served by the BS. By allocating BS power equally across the set of active users\footnote{Power allocation is a thoroughly studied topic in the context of MIMO in general (see, e.g.~ \cite{ChiHan08}). With large antenna arrays, massive MIMO systems are able to get substantially better SINRs even without considering any power allocation optimization. For example,  \cite{HuhTul12} considers equal power allocation while Marzetta in his pioneering work \cite{Mar10} allocates power to users proportional to their channel gains, which is in the reverse direction of typical power allocation  in the context of a fairness criterion (whereby more power is typically allocated at the cell-edge). Following this trend in massive MIMO and considering the high complexity of power allocation optimization, we consider equal active user-stream power allocation at each BS. This approach simplifies the parametrization of peak-rate calculations in Section \ref{sec:analyze} and yields the convex NUM formulation in Section \ref{sec:formulation}.} and fixing the scheduling set sizes for BSs belong to common BS clusters, the instantaneous rate can be predicted \emph{a priori} and independently of the other active-user set, thereby substantially reducing the complexity of the resource allocation problem. 

In this work, we restrict our attention to a small set of predefined possible scheduling set sizes, while how to select these sizes depend on many factors (e.g., number of antennas at the BSs, network deployment, etc.) and we leave the study to future work.
Due to the fact that an uplink user-pilot trains all antennas at nearby BSs, we assume BSs that serving users at larger-size clusters (during a given scheduling slot) can schedule a larger number of users in the same  slot than  BSs serving users at smaller-size clusters or with cellular transmission. We thus let the size of scheduling set of any BS $j$ depend on the size of the cluster including BS $j$. To avoid ambiguity on the scheduling set size of different BSs in the same BS cluster, we enforce the following condition: All users served by a given BS $j$ are served in clusters of the {\em same} size on a given RB. This constraint ensures that overlapping clusters on any RB are of the same size.
Let $S_j$ denote the number of users that can be simultaneously served by BS~$j$ in the cellular transmission. We summarized the properties of the specific form of LJT considered in this work in the following definition:


\begin{defi}\label{admissible-dmimo}{\bf Admissible Local Joint Transmission Schemes (ALJTSs):}
An ALJTS schedules users for transmission on a sequence of RBs, and satisfies the following on each RB:
\begin{enumerate}[(1)]
\item All users served by a given BS $j$ are served in BS clusters of the {\em same} size $\Csize$, for some $\Csize\ge 1$;
\item BS $j$ in BS clusters of size $L$ serves at most $S_j(\Csize)$ users with $ S_j \le S_j(\Csize)\le \Csize S_j$ and $M_j \gg S_j(\Csize)$\footnote{In this work, the $S_j$'s and $S_j(L)$'s are assumed to be predefined parameters. Their impact on performance (and their choice) in practice depends on many factors (e.g., available UL pilot dimensions per RB for training in the network, spatial pilot reuse, network deployment, etc.).};
\item The user beams (i.e., precoding vectors) at BS $j$ are designed as if BS $j$ were performing cellular multi-user (MU)-MIMO transmission over (at most) $S_j(L)$ users;
\item All BSs serving an active user transmit the same coded data stream to the user. Each BS transmits this user stream on a beam that is (independently) designed for this user at that BS.
\item Any BS $j$ share its transmit power $P_j$ equally among the users it serves.
\end{enumerate}
\end{defi}

As explained in Sec. \ref{sec:analyze}, the properties of ALJTS enable us to decouple the problems scheduling and load balancing, since the instantaneous active-user rates depend only on the serving BS cluster and are predictable \emph{a priori} based on the given $S_j(L)$ value.

To illustrate these principles, we give LJT examples that obey these rules in Table \ref{sample-RBs}, involving clusters of sizes 1 (i.e., cellular transmission) and 2.
Four BSs are considered with $P_j\!=\!1$, $S_j(1)\!=\!2$ and $S_j(2)\!=\!3$.  As the table reveals, each BS on RB \#1 engages in cellular transmission. 
On RB \#2, pairs of BSs perform LJT with each BS pair serving a triplet of users. RBs \#3 and \#4 provide additional, more interesting modes.
On RB \#3, no user is served by the same cluster. On RB \#4,  BSs 1 and 2 serve users in clusters of size 2, while BSs 3 and 4 serve users in cellular transmission. Note that if orthogonal pilots are used, (at least) 8, 6, 6 and 7  uplink pilot dimensions (one dimension per user) are needed to enable RBs \#1, \#2, \#3 and \#4, respectively. Evidently, the choice of scheduled user sizes $S_j(\Csize)$ signifies how aggressively pilot dimensions are reused across the network (e.g., $S_j$ for fully reused pilots and $\Csize S_j$ for orthogonal pilots). Inspection of Table \ref{sample-RBs} reveals that the first ALJTS property can be satisfied either if all clusters at an RB are of the same size (RBs \#1,2,3) or if different size clusters are disjoint across BSs (RB \#4).


\begin{table}
\caption{Example of RBs enabled by LJT over 4 BSs. } 
\begin{center}
\begin{tabular}{|c||c|c|c|c|c|}
\hline
RB  & & BS 1 & BS 2& BS 3 & BS 4\\
\hline 
\hline
 & Cluster Size & 1 & 1 & 1 & 1\\
\hline
\#1 & User Power & 1/2 & 1/2 &  1/2 & 1/2\\
\hline 
& Served Users & 1,2 & 3,4 & 5,6 & 7,8 \\
\hline
\hline
 & Cluster Size & 2 & 2 & 2 & 2\\
\hline
\#2 & User Power & 1/3 & 1/3 &  1/3 & 1/3\\
\hline 
& Served Users & 1,2,3& 1,2,3 & 4,5,6& 4,5,6\\
\hline
\hline
 & Cluster Size & 2 & 2 & 2 & 2\\
\hline
\#3 & User Power & 1/3 & 1/3 &  1/3 & 1/3\\
\hline 
& Served Users & 1,2,3 & 1,4,5 & 2,4,6 & 3,5,6 \\
\hline
\hline
 & Cluster Size & 2 & 2 & 1 & 1\\
\hline
\#4 & User Power & 1/3 & 1/3 &  1/2 & 1/2\\
\hline 
& Served Users & 1,2,3 & 1,2,3 & 4,5 & 6,7\\
\hline
\end{tabular}
\end{center}
\label{sample-RBs}
\end{table}

Depending on the availability of bands and the preferences of the operator, in 5G HetNets, different groups of BSs might be allowed to jointly transmit in clusters across groups of RBs (e.g., across different frequency bands). Each such combination of BS clusters is considered separately as a distinct entity in the resource allocation problem. Given the set of BSs $\Bc$ in the network, there are $2^{|\Bc|}-1$ many possible different entity/operation options. Although in principle it is straightforward to take into account all these different options, for simplicity and considering the general interest in 5G HetNet deployments, we focus only on the following options: 1) macro and small cell layers may operate together; 2) only the macro layer operates; 3) only the small cell layer operates. We call these 3 different operations as \emph{shared, macro-only} and \emph{blanking operations}, respectively. Let $A\!\in\!\{1,2,3\}$ denote the operation type, where $A\!=\!1$, $A\!=\!2$ and $A\!=\!3$ denote shared, macro-only and blanking operations, respectively. Let RBs allocated to operation $A$ form Band-$A$, and the set of BSs that can transmit in Band-$A$ be $\Bc^{(\!A)}$. We have the following cases: \\
1) $A\!=\!1$:  shared operation is considered for this band, where macro and small cell BSs can both transmit, i.e., $\Bc^{(1)}\!=\!\Bc$. In this band, clusters can be formed by BSs from different layers. \\
2) $A\!=\!2$: macro-only operation is considered for this band, and only macros can transmit, i.e., $\Bc^{(2)}\!=\!\Bc_m$.  \\
3) $A\!=\!3$: blanking is applied to this band, where all of the macro BSs are muted, i.e., $\Bc^{(3)}\!=\!\Bc_s.$\footnote{In general, blanking can be applied to only a subset of of macro BSs.  In fact, \cite{LiuLau14} has shown gains compared to blanking all macro BSs. Our formulation can include partial blanking by considering new operation options. To simplify the exposition, in this work we confine ourselves to blanking all of the macro BSs.}
 
Different resource allocations among operations refer to different scenarios. For example, scenarios with $A\!\in\!\{2,3\}$ correspond to cases where orthogonal RBs are allocated to macro and small cells, while scenarios with $A\!\in\!\{1,3\}$ can be applied to cases with eICIC.  In some scenarios, resource allocation among operations can be fixed \textit{a priori}. In more flexible scenarios, resource allocation among operations can be a part of the optimization problem. Our formulation in Sec. \ref{sec:formulation} can be applied to both cases.

To show how these different transmission operations of practical interest can be considered within our specific LJT, we give a small example in Table \ref{sample-RBs-blanking}. 
As different bands may prefer different cluster sizes, we consider clusters up to size $L_{\max}^{(A)}$ in Band-$A$. Then the potential BS clusters that can be active in this band is given by the subsets of $\Bc^{(A)}$ with size less than or equal to $L_{\max}^{(A)}$.\footnote{Many of these clusters are not necessary. For example, clusters between BSs that are geographically distant are not necessary to consider, as in a practical system no user would be assigned to these clusters. This type of practical observations can eliminate many cluster options for all RBs. In the following, to avoid cumbersome notation, while listing potential clusters, we only use the cluster size and the active BS set of the corresponding band to describe potential BS clusters.}
Table \ref{sample-RBs-blanking} is extended from Table \ref{sample-RBs} by adding the macro BS (BS \#5) and considering other BSs as small cell BSs, i.e., $\Bc_m\!=\!\{5\}$ and $\Bc_s\!=\!\{1,2,3,4\}$. Assume $L_{\max}^{(1)}\!=\!2, L_{\max}^{(2)}\!=\!1$ and $L_{\max}^{(3)}\!=\!2$ in this example. 
RBs  $b$ and $d$ are in Band-$1$ (shared operation), RB $e$ is in Band-$2$ (macro only operation), and RBs $a$ and  $c$ are in Band-$3$ (blanking operation). 
In RB $b$, each BS is performing cellular transmission, while RB $d$ considers clusters of size 2 including both macro and small cell BSs.  In RB $e$, the macro BS (BS \#5) serves users via cellular transmission. 
In RBs $a$ and $c$,  only small cells serve users while the macro BS is muted. In fact, clusters in RBs $a$ and $c$ are the same as in  RBs \#1 and \#3 in Table \ref{sample-RBs}, respectively. 

\begin{table}
\caption{Example of RBs enabled by ALJTS over 4 small cell BSs  (BSs 1-4) and 1 macro BS (BS 5). } 
\begin{center}
\begin{tabular}{|c||c|c|c|c|c|c|}
\hline
RB  & & BS 1 & BS 2& BS 3 & BS 4& (Macro) BS 5\\
\hline 
\hline
 & Cluster Size & 1 & 1 & 1 & 1&-\\
\hline
$a$ & User Power & 1/2 & 1/2 &  1/2 & 1/2&-\\
\hline 
& Served Users & 1,2 & 3,4 & 5,6 & 7,8 &-\\
\hline
\hline
 & Cluster Size & 1 & 1 & 1 & 1&1\\
\hline
$b$ & User Power & 1/2 & 1/2 &  1/2 & 1/2&1/2\\
\hline 
& Served Users & 1,2 & 3,4 & 5,6 & 7,8 &9,10\\
\hline
\hline
 & Cluster Size & 2 & 2 & 2 & 2&-\\
\hline
$c$ & User Power & 1/3 & 1/3 &  1/3 & 1/3&-\\
\hline 
& Served Users & 1,2,3 & 1,4,5 & 2,4,6 & 3,5,6& -\\
\hline
\hline
 & Cluster Size & 2 & 2 & 2 & 2&2\\
\hline
$d$ & User Power & 1/3 & 1/3 &  1/3 & 1/2&1/3\\
\hline 
& Served Users & 1,2,3 & 1,4,5 & 2,4,6 & 3,7&5,6,7\\
\hline
\hline
 & Cluster Size & - & - & - &-&1\\
\hline
$e$ & User Power & - & - &  - & -&1/2\\
\hline 
& Served Users &-& - & - & - &1,2\\
\hline
\end{tabular}
\end{center}
\label{sample-RBs-blanking}
\end{table}

\begin{remark} For any band, potential BS clusters are given by subsets (of appropriate size) of the active BS set on that band. As noted earlier, some subsets can be eliminated given the topology and constraints in the network. At the scheduling slot scale, the active BS clusters are determined by the scheduler that is operated by the central controller.  Fig. 2 provides a flow diagram for the interactions between different entities, such as  BSs, users and central controller and includes a load balancing and a scheduling unit. These interactions between different entities and the processing performed within different entities are discussed in the Sections \ref{sec:analyze}, \ref{sec:vq-scheme} and \ref{sec:simulation}. 
\end{remark}

\begin{figure}
\centering
\includegraphics[width=6.59cm, height=4.05cm]{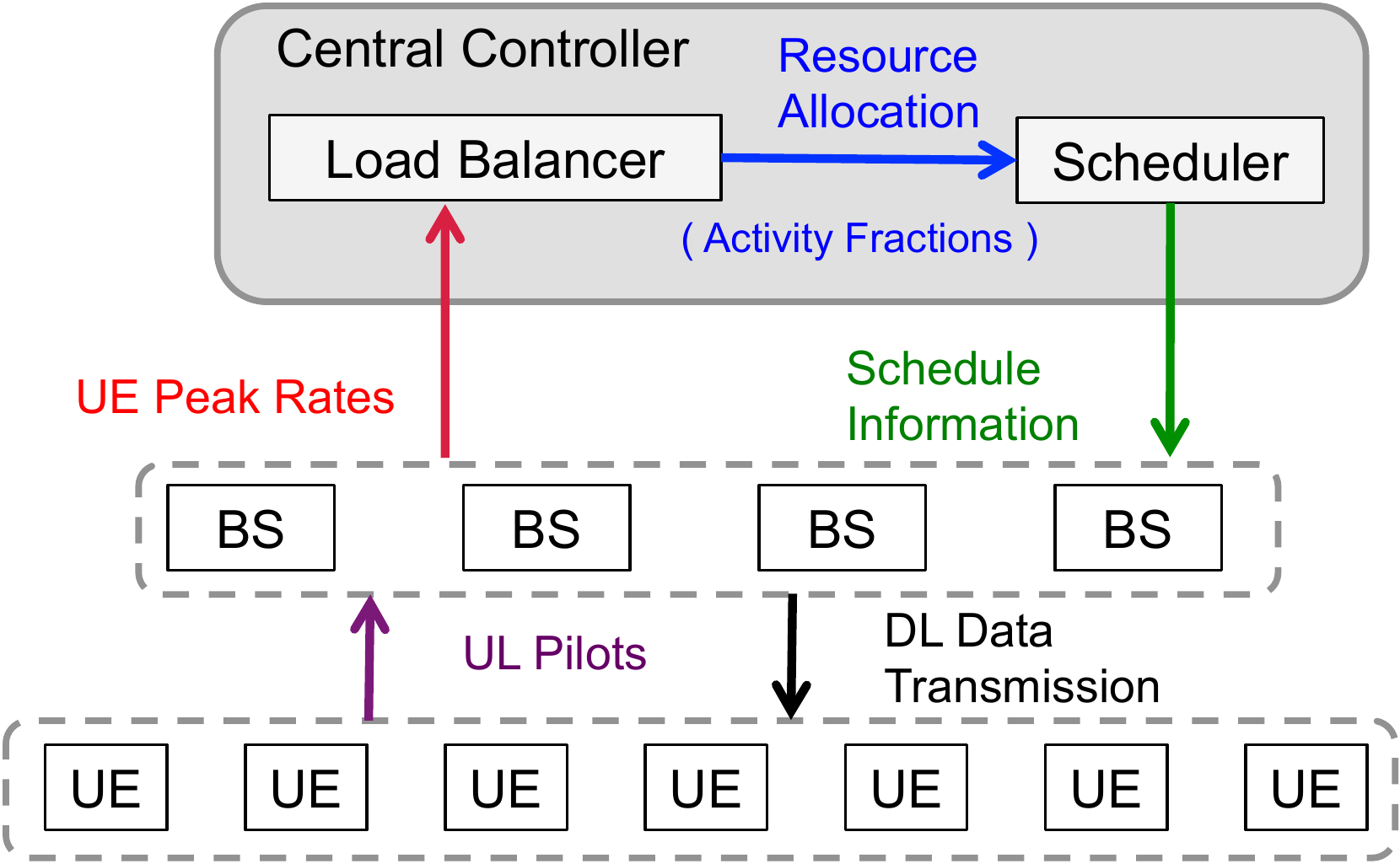}
\caption{Given the resource allocation done by the load balancer at the central controller, the scheduling is done at the scheduler unit in the central controller as shown in Fig. \ref{fig:flow-diagram}.}
\label{fig:flow-diagram}
\end{figure}

\section{Instantaneous and Long-term Rates}\label{sec:analyze}
In this section, we provide proxy expressions for instantaneous rates and long-term rates (throughput)  with either LZFBF or Maximum Ratio Transmission (MRT, also known as Conjugate Beamforming).
We denote the transpose and  conjugate transpose of matrices by $(\cdot)^T$ and $(\cdot)^H$, respectively. 

On a generic RB $t$, we let $g_{kj,i}=\sqrt{\beta_{kj}}h_{kj,i}$ be the channel between the $i^{\rm th}$ transmit antenna of BS $j$ and any user $k$, and it includes both slow fading $\beta_{kj}$ and fast fading~$h_{kj,i}$. The slow fading $\beta_{kj}$ characterizes the combined effect of  distance-based path loss and  location-based shadowing. Let $\Kc_j(t)$ be the set of users served by BS $j$ and $S_j(t) = |\Kc_j(t)|$ denote the number of users scheduled by BS $j$ on RB $t$. The channel matrix between BS $j$ and its active users (the users scheduled by this BS at this RB) is denoted by $\mathbf{G}_j$, where the dimension of $\mathbf{G}_j$ is $M_j\!\times\!S_j(t)$.  The $m^{\rm th}$ column of $\mathbf{G}_j$ 
corresponds to the channel of user $k = k_j(m)$ for some $k\in\{1,2,\ldots,\Kc\}$. That is, for a given $m$ we have 
$[\mathbf{G}_j]_{im} =  g_{kj,i}$, with $k = k_j(m)$.   This expression can also be interpreted in terms of the inverse mapping  $m=m_j(k)$, that is,  for a given $m$ we have 
$[\mathbf{G}_j]_{im} =  g_{kj,i}$, with  $m = m_j(k)$.

We assume each link experiences independent Rayleigh fading, i.e., $\mathbf{h}_{kj}=[h_{kj,1}, \cdots, h_{kj,M_j}]^T$ are complex Gaussian i.i.d. random variables.\footnote{The instantaneous user-rate expressions in this paper assume no spatial correlation in the user channels.  In principle, instantaneous user-rate expressions can be developed for spatially correlated user channels with a given spatial correlation structure by using the method of deterministic equivalent \cite{HoyTen13}. 
The framework presented in Secs.~\ref{sec:formulation}--\ref{sec:vq-scheme}, however,  is not directly applicable with spatially correlated user channels, since a user's instantaneous rate would in general depend on the spatial correlation of the other scheduled user channels. Although  beyond the scope of this paper,  extensions involving spatially correlated user channels is an area worth further investigation.} We let $\mathbf{F}_j$ denote the precoding matrix at BS $j$ with dimension $M_j\!\times\!S_j(t)$, whose $m^{\rm th}$ column $\mathbf{f}_{mj}$ is the beam (i.e., the precoding vector) for the $m^{\rm th}$ user of BS $j$, i.e.,  the user whose channel to the BS is given by the $m$-th column of $\mathbf{G}_j$. 
The signal symbol of user $k$ is denoted by $s_k$, where $s_k$ has unit energy. The thermal noise at user $k$ is denoted by $w_k$, which is assumed to be additive white Gaussian noise (AWGN)  with variance~$\sigma^2$.

We consider a scheduling policy on RBs  $\{1, \, 2\, \cdots, T\}$  and assume that all the large-scale coefficients stay fixed within this period. Any such scheduling policy can be described in terms of  the scheduling sets $\{ \Uc_\Cc(t);  \forall \Cc, \forall t \in \{1, \, 2\, \cdots, T\} \}$, where $\Uc_{\Cc}(t)$ denotes the set of users served by cluster $\Cc$ on RB~$t$ and $\forall \Cc$ denotes all of the possible cluster options considered for all bands taking into account active BS sets, cluster sizes and possible cluster selections/eliminations.
 Without loss of generality, we assume that RB $t$ is in Band-$A$ and $\Cc\!\subset\!\Bc^{(A)}$.
Thus, the received signal at user $k\in \Uc_{\Cc}(t)$ on RB $t$ is

\begin{equation}\label{eq:receivesig-0}
\begin{aligned}
y_{k}(t)= &\underbrace{\sum_{j\in \Cc} \sqrt{\frac{P_{j}}{S_{j}(|\Cc|)}} \mb{g}_{kj}(t)^H \mb{f}_{m_j(k)j}(t)s_k }_\text{desired} + \underbrace{ \sum_{j\in \Cc} \sum_{u\in\cup_{(\Cc':j\in\Cc')}\Uc_{\Cc'}(t), \; u\neq k}\sqrt{\frac{P_{j}}{S_{j}(|\Cc|)}} \mb{g}_{kj}(t)^H \mb{f}_{m_j(u)j}(t)s_u }_\text{intra-cluster interference} \\
&+\underbrace{\sum_{l\notin \Cc} \sum_{u\in \cup_{(\mathcal{C'}: l\in \mathcal{C'})} \Uc_\mathcal{C'}(t)}\sqrt{\frac{P_l}{S_{l}(|\mathcal{C'}|)}} \mb{g}_{kl}(t)^H \mb{f}_{m_\ell(u)l}(t)s_u }_\text{inter-cluster interference} + \underbrace{w_k}_\text{noise}.
\end{aligned}
\end{equation}


Adopting LZFBF, the precoding matrix at BS $j$ is $\mb{F}_j\!=\!\mb{G}_j\left(\mb{G}_j^H\mb{G}_j\right)^{-1}\mb{A}_j^{1/2}$, where $\mb{A}_j$ is the normalizing coefficients matrix. Specifically, $\mathbf{A}_j$ is a $S_j(L) \times S_j(L)$ diagonal matrix with the $k$th diagonal element being $a_{k,k} = \frac{1}{\left[\left( \mathbf{G}_j^H \mathbf{G}_j\right)^{-1}\right]_{k,k}}$, where $[\left( \mathbf{G}_j^H \mathbf{G}_j\right)^{-1}]_{k,k}$ denotes the $k$th row and $k$th column element of the matrix $\left( \mathbf{G}_j^H \mathbf{G}_j\right)^{-1}$. In this case, the intra-cluster interference is $0$. With MRT, the precoding matrix  at BS $j$ is $\mb{F}_j$ with the $m$th column being $\mb{f}_{mj}\!=\!\frac{\mb{g}_{k_j(m)j}}{\|\mb{g}_{k_j(m)j}\|}$.  Note that the set of interfering BSs depends on the operation band. 

\vspace{-0.4cm}
\subsection{Instantaneous Rate}
In this paper, we assume that each BS has  available perfect CSI regarding the user terminals it serves.\footnote{Massive MIMO rate calculations for general CoMP setting  with practical TDD UL training for CSI acquisition can be found in \cite{HuhTul12}. The instantaneous rate expression in this manuscript can be updated accordingly taking into account the pilot contamination term and the training overhead as a special case of \cite{HuhTul12}. This was done for the cellular case in \cite{BetBur14a}.} Let $S_j$ denote the number of users served by BS $j$ in cellular transmission, with $S_j\!\ll\!M_j$. Under mild assumptions on fading,  the user instantaneous rates on RB $t$, $r_{kj}(t)$,  can be predicted \emph{a priori} in the massive MIMO regime \cite{BetBur14a}. In particular, there exist deterministic quantities $\{r_{kj}\}$ such that  $r_{kj}(t) \stackrel{{\rm a.s.}}{\rightarrow} r_{kj}$,  $\forall k\!\in\!\Uc$ and $\forall j\!\in\!\Bc$,  as $M_j, S_j \!\rightarrow\!\infty$, with fixed  $v_j\!=\!S_j/M_j\!\geq\!0$  \cite{Mar10,HoyTen13,HuhCai11}. This convergence  is very fast with respect to $M_j$'s. 
Unlike general CoMP, where a user's instantaneous rate depends on the other users co-scheduled on the same RB  \cite{HuhTul12}, the ALJTS makes a user's instantaneous rate \emph{independent}  of  the other users in the scheduling set. 
Let $r_{k\Cc}^{(A)}(t)$ be the instantaneous rate of user $k$ from cluster $\Cc$ on RB $t$ in Band-$A$. There exist deterministic quantities $\{r_{k\Cc}^{(A)}\}$ such that $r_{k\Cc}^{(A)}(t) \stackrel{{\rm a.s.}}{\rightarrow} r_{k\Cc}^{(A)}$ as $M_j, S_j(|\Cc|)\!\rightarrow\!\infty$ with $v_j\!=\!S_j(|\Cc|)/M_j\!\geq\!0, \forall j\in\Cc$. 

Using the techniques in \cite{LimCha13,HuhTul12}, we can show that the approximate  instantaneous rate of user $k$ from cluster $\Cc\!\subset\!\Bc^{(A)}$ in Band-$A$ using LZFBF is
\begin{equation}\label{eq:rate-zf}
r_{k\Cc}^{(A)}=\log_2\left(1+\frac{\sum_{j\in \Cc} \sum_{l\in \Cc} \sqrt{P_{j}P_{l}\beta_{kj}\beta_{kl}b_{j}(|\Cc|)b_{l}(|\Cc|) }}{\sigma^2 + \sum_{l\notin \Cc, l\in\Bc^{(A)}} P_l \beta_{kl}}\right),
\end{equation}
where $b_{j}(a)=\frac{M_{j} - S_{j}(a)+1}{S_{j}(a)}$.
Similarly, the approximate instantaneous rate  of user $k$ from cluster $\Cc$ in Band-$A$ using MRT is
\begin{equation}\label{eq:rate-mrt}
r_{k\Cc}^{(A)}= \log_2\left(1+ \frac{\sum_{j\in \Cc}\sum_{l\in \Cc}\sqrt{\frac{P_j P_l M_j M_l \beta_{kj}\beta_{kl}}{S_{j}(|\Cc|)S_{l}(|\Cc|)}} }{\sigma^2 + I_{k\Cc}+ \sum_{l\notin \Cc, l\in \Bc^{(A)}} P_l\beta_{kl}} \right),
\end{equation}
where $I_{k\Cc}=\sum_{j\in \Cc} \frac{(S_{j}(|\Cc|)-1)}{S_{j}(|\Cc|)} P_j\beta_{kj} $ is the non-zero intra-cluster interference. 
Clearly, $r_{k\Cc}^{(A)}\!=\!0$ if $\Cc\!\not\subset\!\Bc^{(A)}$.


Eqs.~(\ref{eq:rate-zf}) and (\ref{eq:rate-mrt}) assume that  $\forall j\!\in\!\Cc$, BS $j$ serves $S_j(|\Cc|)$ users and allocates $P_j/S_j(|\Cc|)$ fraction of its power to each user. In the case that fewer users are served  by one of the BSs, (\ref{eq:rate-zf}) and (\ref{eq:rate-mrt})  represent achievable lower-bound instantaneous rates. 


It is worth noting that LJT provides instantaneous-rate improvements at the cell-edge with respect to cellular, as it replaces the  intra-cell BF gain provided by the cellular schemes with intra-cluster BF gain. Also, given that macro BSs do not transmit in Band-$3$ (i.e., the blanking operation), 
users in small cells benefit from larger SINRs in Band-$3$, as there is no interference from macro BSs to small cell users in this band. 

\vspace{-0.2cm}
\subsection{Long-term Rates}\label{sec:longtermR}
\vspace{-0.2cm}
As discussed in \cite{BetBur14a},  users can be served (at distinct scheduling instances) by more than one BS in massive MIMO networks in cellular transmission.
Similarly, users can be served by different clusters on different RBs in ALJTS. 
%
Let  $x_{k\Cc}^{(A)}\!=\!\lim_{T \rightarrow \infty} \frac{| \{ t: 1\leq t\leq T, \ k \in \Uc_{\Cc}^{(A)}(t), \  t  \textrm{ in Band-}A \}| }{T}$ denote the activity fraction of user $k$ to cluster $\Cc$ in Band-$A$, that is, the fraction of resources allocated by cluster $\Cc$ to user $k$ in  Band-$A$. The activity fraction $x_{k\Cc}^{(A)}$ is a real number showing the fraction of RBs (averaged over many slots within which the load balancing is considered) where user-$k$ is served by cluster-$\Cc$ in Band-$A$.  If $x_{k\Cc}^{(A)}\!>\!0$,  user $k$ is served by cluster $\Cc$ in Band-$A$. We obtain the long-term rate similar to  \cite{BetBur14a}, using instantaneous rates and  activity fractions from the scheduling policy.  In particular, in the limit $T\to \infty$, the long-term rate of user $k$  equals\footnote{Convergence to the limiting expressions of interest is very quick \cite{BetBur14a}.}
\begin{equation}
\label{avgtp-vs-actfracs}
R_k = \sum_{A=1}^{3}\sum_{
\substack{
\Cc: \Cc\subset\Bc^{(A)} \\
|\Cc|\leq L_{\max}^{(A)}}
} x_{k\Cc}^{(A)}\, r_{k\Cc}^{(A)}.
\end{equation}

\section{Unified NUM Problem Formulation}\label{sec:formulation}
While ALJTS allows for varying cluster sizes within an RB as revealed in Table \ref{sample-RBs}, in the sequel we specialize to the tractable case of practical interest involving equal-size clusters on each RB. Consequently, the following framework allows for cluster options as seen in RBs \#1-3 but not for RB \#4 in Table \ref{sample-RBs}. We call this new, additionally constrained, scheme the Uniform Cluster-Size scheme (UCS).\footnote{ Unlike UCS, when different size clusters are allowed to operate together, great care in the scheduling design must be given to avoid overlapping clusters with different sizes to operate in the same RB. This makes scheduling very complicated. The considered UCS provides a lower bound on the performance compared to more general ALJTS, which can serve a useful benchmark.}
\begin{defi}\label{first-arch}
{\bf Uniform Cluster-Size Scheme (UCS)}: 
ALJTS is a UCS if 
\begin{enumerate}[(1)]
\item $\mu_A$ fraction of RBs is allocated to Band-$A$, with $\sum_{A} \mu_A\leq 1$;
\item For each Band-$A$, $\lambda_{AL}$ fraction of RBs is allocated to  size-$L$ clusters for $1\leq L\leq L_{\max}^{(A)}$, with $\sum_{L=1}^{L_{\max}^{(A)}} \lambda_{AL}\leq \mu_A$;
\item on any RB in the  $\lambda_{AL}$ fraction,  the scheduled users are served by  (user-dependent) clusters of the same size $L$ and these clusters are formed by BSs in $\Bc^{(A)}$;
\item on any RB in the  $\lambda_{AL}$ fraction, each BS does not serve more than $S_j(L)$ users.
\end{enumerate}
\end{defi}
Then LJT designs considered in the rest of the paper are all \emph{Uniform Cluster-Size Schemes}. RBs allocated to serving size-$L$ clusters in Band-$A$ comprise what we call the $L^{\rm th}$ subband of Band-$A$. The NUM problem for the UCS optimizes activity fractions and subband/band  allocations is as follows:
\begin{subequations}\label{eq:opt-cvx}
\begin{align}
\max\limits_{\lambda_{A\Csize}, x_{k\Cc}^{(A)},\mu_A} \ & \sum_{k\in\mathcal{U}} U\left( \sum_{A=1}^{3} \sum_{\substack{
\Cc: \Cc\subset \Bc^{(A)},  \\ |\Cc|\leq L_{\max}^{(A)}}} x_{k\Cc}^{(A)}r_{k\Cc}^{(A)}\right)\\
\text{s.t. } &\sum_{\substack{\Cc: \Cc\subset \Bc^{(A)}, \\ j\in \Cc, |\Cc|=\Csize}}\frac{\sum_{k\in\mathcal{U}} x_{k\Cc}^{(A)}}{S_j(\Csize)} \leq \lambda_{AL}, \ \forall j\in \Bc^{(A)}, \forall L \leq L_{\max}^{(A)}, \forall A, \label{eq:opt-ct-cluster-cvx}\\
& \sum_{\Cc: |\Cc| = L, \Cc\subset \Bc^{(A)}} x_{k\Cc}^{(A)} \leq \lambda_{AL}, \ \forall k \in \Uc, \forall L \leq L_{\max}^{(A)}, \forall A, \label{eq:opt-ct-ue-cvx}\\
&\sum_{L=1}^{L_{\max}^{(A)}} \lambda_{AL} \leq \mu_A, \forall A,\label{eq:opt-ct-sumArch-cvx}\\
&\sum_{A=1}^3 \mu_A \leq 1,\label{eq:opt-ct-mu-cvx}\\
& x_{k\Cc}^{(A)}, \lambda_{AL}, \mu_A \geq 0, \ \forall k \in \Uc, \forall \Cc, \forall L\leq L_{\max}^{(A)}, \forall A,\label{eq:opt-ct-positive-cvx}
\end{align}
\end{subequations}
where the utility function $U(\cdot)$ is a continuously differentiable, monotonically increasing, and strictly concave function \cite{StaWic09}. Constraint (\ref{eq:opt-ct-cluster-cvx})  signifies that the total  activity fractions allocated by BS $j$ in clusters of size $\Csize$ in Band-$A$ cannot exceed the total available resources $\lambda_{A\Csize} S_j(\Csize)$. 
On the other hand, recalling that each user cannot be served by multiple clusters on the same RBs, (\ref{eq:opt-ct-ue-cvx})  signifies that the fraction of RBs over which user $k$ is served by clusters of size $\Csize$ in Band-$A$ cannot exceed RBs allocated to the  clusters of size $\Csize$ in Band-$A$, $\lambda_{A\Csize}$. (\ref{eq:opt-ct-sumArch-cvx}) ensures that the total resources allocated to the subbands in Band-$A$ are no more than the resources allocated to that band. Finally, (\ref{eq:opt-ct-mu-cvx}) signifies the fact that the summation of resources allocated to different bands is equal to all available resources.

\begin{remark} 
The formulation (\ref{eq:opt-cvx}) is quite flexible. It can be applied not only to the scenarios that optimize the resource allocation among operations, but also to the scenarios where resources given to each operation are fixed \emph{a priori} by setting the corresponding $\mu_A$ values to constants in (\ref{eq:opt-cvx}). The cellular transmission \cite{BetBur14a} can be recovered as a special case of (\ref{eq:opt-cvx}) by setting $A = 1$ and $L_{\max}(1) = 1$.
\end{remark}

Any concave function (e.g., general $\alpha$-fairness, \cite{MoWal00}) can be applied to the problem formulation (\ref{eq:opt-cvx}). The formulated optimization is convex as long as the utility function is concave \cite{Ber09}. In this paper, among various concave utility functions available in the literature, we work with the logarithmic utility which is also known as ``proportional fairness'' \cite{YeAnd13LB,DebMon14,BetBur14a}.\footnote{There are many options for the utility function such as (weighted) arithmetic mean, geometric mean, max-min fairness. Each option can be relevant for certain scenario and interest. Both arithmetic mean and max-min fairness have their shortcomings \cite[Chapter 1]{Bjornson13}. Geometric mean (aka proportional fairness), promotes a trade-off of user-rates between the other two utility functions. The ``$\log$'' function in the utility ensures diminishing returns for individual users rates as they get higher. This de-motivates the optimization to give high rates to a few users. } General numerical solvers (e.g., CVX) can be used. Since CVX is not well-suited for large instances~\cite{cvx}, we alternatively propose an efficient algorithm that can be applied to large networks in the next section and the complexity difference between the proposed algorithm and a general numerical solver is investigated in Appendix \ref{app:complexity}.

\section{Dual Subgradient Based Algorithm}\label{sec:dual-algo}
In this section, we propose an efficient algorithm based on the dual subgradient method \cite{Ber09}. 
We let $\nu_{jL}^{(A)}$ and $\theta_{kL}^{(A)}$ be the Lagrange multipliers corresponding  to (\ref{eq:opt-ct-cluster-cvx}) and (\ref{eq:opt-ct-ue-cvx}), respectively. 
The dual problem of (\ref{eq:opt-cvx}) is
$\min\limits_{\nu_{jL}^{(A)}, \theta_{kL}^{(A)} \geq 0} \ \sum_{k\in\mathcal{U}} f_k\left(\nu_{jL}^{(A)}, \theta_{kL}^{(A)}\right) + g\left(\nu_{jL}^{(A)}, \theta_{kL}^{(A)}\right),$
where 
\begin{equation}\label{eq:opt-dual-UE}
\begin{aligned}
f_k\left(\nu_{jL}^{(A)}, \theta_{kL}^{(A)}\right) =\max\limits_{x_{k\Cc}^{(A)}\geq 0} \  &\log\left(\sum_{A=1}^3 \sum_{\substack{
\Cc: \Cc\subset \Bc^{(A)},\\ |C|\leq L_{\max}^{(A)}}}x_{k\Cc}^{(A)}r_{k\Cc}^{(A)} \right) - \sum_{A=1}^3\sum_{L=1}^{L_{\max}^{(A)}}\sum_{\substack{\Cc:\Cc\subset\Bc^{(A)},\\ |\Cc|=L}}\sum_{j: j\in\Cc} \frac{\nu_{jL}^{(A)}}{S_j(L)}x_{k\Cc}^{(A)}\\
& - \sum_{A=1}^3\sum_{L=1}^{L_{\max}^{(A)}}\theta_{kL}^{(A)} \sum_{\Cc:\Cc\subset\Bc^{(A)}, |\Cc|=L} x_{k\Cc}^{(A)},
\end{aligned}
\end{equation}
and 
\begin{equation}\label{eq:opt-dual-lambda}
g(\nu_{jL}^{(A)}, \theta_{kL}^{(A)}) =
\max\limits_{\substack{\sum_{L=1}^{L_{\max}^{(A)}}\lambda_{AL}\leq \mu_A,\\ \sum_{A=1}^3 \mu_A\leq 1} } \sum_{A=1}^3\sum_{L=1}^{L_{\max}^{(A)}}\left( \sum_{j: j\in \Bc^{(A)}} \nu_{jL}^{(A)}  + \sum_{k\in\Uc}\theta_{kL}^{(A)}\right)\lambda_{AL}.
\end{equation}
The constraints of (\ref{eq:opt-cvx}) satisfy the Slater condition \cite{Ber09}, and thus strong duality holds (i.e., the dual problem  and the original problem (\ref{eq:opt-cvx}) have the same optimal value). 

\vspace{-0.4cm}
\subsection{The Dual Subgradient Method}
The optimization problem (\ref{eq:opt-dual-UE}) has the closed-form optimal solution
\begin{equation}\label{eq:dual-x}
x_{k\Cc}^{(A)} =
\begin{cases}
\frac{1}{\sum_{L:L=|\Cc|}\left(\sum_{j:j\in\Cc}\nu_{jL}^{(A)}/S_j(L)+\theta_{kL}^{(A)}\right)}, & \text{if } \{\Cc,A\} = \{\Cc^*, A^*\},\\
0, & \text{otherwise},
\end{cases}
\end{equation}
where $\{\Cc^*, A^*\}=\arg\max_{\Cc, A}  \frac{r_{k\Cc}^{(A)}}{\sum_{L:L=|\Cc|}\left(\sum_{j:j\in\Cc}\nu_{jL}^{(A)}/S_j(L)+\theta_{kL}^{(A)}\right)}.$\footnote{If we have multiple pairs of $\{\Cc^*,A^*\}$, we just randomly pick one pair.}

The problem (\ref{eq:opt-dual-lambda}) is an LP and one optimal solution is\footnote{If we have multiple $\{A,L\}$ pairs that maximize the $\sum_{j: j\in \Bc(A)} \nu_{jL}^{(A)}  + \sum_{k\in\Uc}\theta_{kL}^{(A)}$, we just randomly pick one.}
\begin{equation}\label{eq:lambdaopt-dualalgo}
\lambda_{AL} = \left\{
\begin{aligned}
& 1, \text{ if } \{A,L\}= \arg\max_{A',L'} \sum_{j: j\in \Bc(A')} \nu_{jL'}^{(A')}  + \sum_{k\in\Uc}\theta_{kL'}^{(A')}, \\
&0, \text{ otherwise},
\end{aligned}\right.
\end{equation}
\begin{equation}\label{eq:muopt-dualalgo}
\mu_A = \left\{
\begin{aligned}
& 1, \text{ if there exists a band $A$ such that the above } \lambda_{AL}>0, \\
&0, \text{ otherwise}.
\end{aligned}\right.
\end{equation}

The $t$th iteration of the  algorithm is as follows.
\begin{enumerate}
\item Update the activity fractions by (\ref{eq:dual-x}).
\item  Update resource allocation for different bands and clusters by (\ref{eq:lambdaopt-dualalgo}) and (\ref{eq:muopt-dualalgo}). 
\item Update the Lagrangian multipliers by 
\begin{equation}\label{eq:dual-nu}
\nu_{jL}^{(A)}(n+1) = \left[\nu_{jL}^{(A)}(n) - \delta(n) \left(\lambda_{AL}(n) - \sum_{\substack{\Cc: \Cc\subset \Bc(A),\\ j\in \Cc, |\Cc|=L}}\frac{\sum_{k\in\mathcal{U}} x_{k\Cc}^{(A)}(n) }{S_{j}(\Csize)}\right)\right]^+,
\end{equation}
and 
\begin{equation}\label{eq:dual-theta}
\theta_{kL}^{(A)}(n+1) = \theta_{kL}^{(A)}(n) - \delta(n) \left(\lambda_{AL}(n) - \sum_{\Cc: \Cc\subset\Bc(A), |\Cc|=L} x_{k\Cc}^{(A)} \right),
\end{equation}
where  $[z]^+=\max\{z, 0\}$ and $\delta(n)$ is the stepsize at the $n^{\rm th}$ iteration.
\end{enumerate}

By adding redundant constraints $x_{k\Cc}^{(A)}\!\leq\! 1$ and choosing an appropriate stepsize (e.g, a diminishing stepsize $\delta(n)\!=\!\frac{a}{n+b}$, where $a$ and $b$ are some positive scalars), the subgradients can be bounded. This allows us leveraging Prop. 6.3.4. in \cite{Ber09} to show the convergence of the dual subgradient algorithm.  The detailed steps for the algorithm with redundant constraints can be found in Appendix \ref{pf:dual-algo}. 


\vspace{-0.4cm}
\subsection{Finding the Optimal Primal Solutions Given the Optimal Dual Variables} \label{sec:dual-to-primal}
Note that the objective function of (\ref{eq:opt-cvx}) is not strictly convex and we may have multiple optimal solutions. In this case, given the optimal dual variables, it is generally difficult to find the optimal primal solutions that satisfy the KKT conditions. However, by exploring the structure of (\ref{eq:opt-cvx}) as follows, we propose to obtain the optimal primal solutions by solving a small-size LP. 

The optimal long-term rate $R_k^*\!=\!\sum_{A=1}^3\sum_{\Cc:\Cc\subset\Bc^{(A)}}x_{k\Cc}^{*(A)}r_{k\Cc}^{(A)}$ in (\ref{eq:opt-cvx}) is unique, since the function $\log(R_k)$ is strictly concave with respect to $R_k$. KKT conditions of problem (\ref{eq:opt-cvx}) imply
\begin{equation}\label{eq:kkt-Rk}
R_k \geq \frac{r_{k\Cc}^{(A)}}{\sum_{j:j\in\Cc}\nu_{j|\Cc|}^{(A)}/S_j(|\Cc|)+\theta_{k|\Cc|}^{(A)}}.
\end{equation}
Thus,  given the optimal dual variables, the unique optimal rate can be easily  obtained by  $R_k^*=\max_{\Cc,A} \left\{\frac{r_{k\Cc}^{(A)}}{\sum_{j:j\in\Cc}\nu_{j|\Cc|}^{(A)}/S_j(|\Cc|)+\theta_{k|\Cc|}^{(A)}}\right\}$. We observe from (\ref{eq:kkt-Rk}) that  in the optimal solutions, each user only has positive activity fractions $x_{k\Cc}^{(A)}$ to  clusters providing the maximum term of the right-hand side of (\ref{eq:kkt-Rk}). 
Based on this conclusion, we propose the following LP, whose size is reduced by only focusing on the positive $x_{k\Cc}^{(A)}$ obtained from~(\ref{eq:kkt-Rk}).
\begin{equation}\label{eq:opt-LP}
\begin{aligned}
\max_{\eta, x, \lambda} \ \   & \eta \\
\text{s.t. } &  \eta \leq   \sum_{A=1}^3 \sum_{\Cc\subset\Bc^{(A)}} \frac{ x_{k\Cc}^{(A)} r_{k\Cc}^{(A)} }{R_k^*},\ \forall k\in\mathcal{U},\\
&(\ref{eq:opt-ct-cluster-cvx})-(\ref{eq:opt-ct-positive-cvx}).
\end{aligned}
\end{equation}
\begin{prop}\label{prop:primal-optimal}
Given that $R_k^*$ is the exact optimal rate of (\ref{eq:opt-cvx}), the solution of (\ref{eq:opt-LP}) is the same as the optimal solution of problem (\ref{eq:opt-cvx}). 
\end{prop}
\begin{proof}
Similar techniques in the proof of Lemma 1 in \cite{BetBur14a} can be used to complete this proof.
\end{proof}
Prop. \ref{prop:primal-optimal} implies that we can obtain the solutions of (\ref{eq:opt-cvx}) given the optimal dual variables. Though we can show the convergence of the dual subgradient algorithm by adding redundant constraints, there may exist a small gap between the obtained dual variables and the optimal ones, due to the numerical precision or the limit on the number of iterations. Exploiting the well-behaved structure of~(\ref{eq:opt-LP}), i.e., finite coefficients and a bounded feasible set \cite{BetBur14a}, it is expected that the solution of (\ref{eq:opt-LP}) is near optimal in the presence of a small gap between the obtained dual variables and the optimal ones.

Empirical evidence reveals that in a heavily loaded network, where constraints~(\ref{eq:opt-ct-ue-cvx}) are  inactive (i.e., $\sum_{ \substack{\Cc: |\Cc| = L, \\ \Cc\subset \Bc^{(A)}}} x_{k\Cc}^{(A)}\!<\!\lambda_{AL}$), most users are uniquely served by one cluster on each subband. 
Insight regarding this observation can be obtained by examining KKT conditions of (\ref{eq:opt-cvx}) as follows.

\begin{prop}\label{prop:uniqueass}
For a given Band-$A$ and a cluster size $L$, if (\ref{eq:opt-ct-ue-cvx}) are inactive $\forall k\in\mathcal{U}$, the number of users that are served by multiple BS clusters on RBs allocated to $L^{\rm th}$ subband of Band-$A$ is at most $N_{\Cc L}^{(A)}\!-\!1$, where $N_{\Cc L}^{(A)}$ is the number of clusters in the $L^{\rm th}$ subband of Band-$A$.
\end{prop}
\begin{proof}
See Appendix~\ref{pf:prop-uniqueass}.
\end{proof}

Prop. \ref{prop:uniqueass} implies that the optimal user associations in each subband are mostly unique.  We call the users served by more than one cluster on any subband as ``fractional users''. Note that Prop. \ref{prop:uniqueass} provides an upper bound (i.e., $N_{\Cc L}^{(A)}$) on the number of fractional users, while simulations show a much smaller number of fractional users (less than $3.5\% K$ in Sec. \ref{sec:simulation}). 
Recall that the dual subgradient algorithm determines the set of positive $x_{k\Cc}^{(A)}$ of users to their cluster-band pairs $\{\Cc^*,A^*\}$,  
while the rest of activity fractions are zero. Thus, unknown activity fractions that needs to be solved via  (\ref{eq:opt-LP}) are only the positive activity fractions. Based on Prop. \ref{prop:uniqueass}, most users (with unique association) have at most one positive activity fraction on any subband. 
Thus, the size of (\ref{eq:opt-LP}) is significantly reduced, implying the efficiency of the proposed algorithm. Further details on the algorithm complexity can be found in Appendix \ref{app:complexity}.

In summary, Proposition 1 has revealed the optimality of the specific method proposed in Sec. \ref{sec:dual-to-primal} to obtain the primal variables given the dual variables. Furthermore, the analysis of the number of iterations required for convergence of the proposed algorithm and its complexity reveal the efficiency of the proposed  algorithm with respect to its application to large network instances. Unlike the cellular case \cite{BetBur14a}, it is not \emph{a priori} known whether the NUM solution can be implemented via any scheduler or not. The implementation of NUM solutions is discussed below. 
\section{Scheduling}\label{sec:vq-scheme}
In this section, we develop scheduling policies that yield activity fractions closely matching the NUM solution. Scheduling is done independently and in parallel for each band. As seen in Fig.~\ref{fig:flow-diagram}, a scheduler at a central controller can collect the needed scheduling information and schedule the users according to the proposed scheduling scheme independently for each band (i.e., for each operation option). Considering a scheduling policy for Band-$A$ and letting $L(t)$ be the cluster size  in RB $t$, we define the feasible scheduling policy as follows.
\begin{defi}\label{feas-sched-one} {\bf Feasible Schedule:}
A scheduling policy 
$\left\{ \Uc_\Cc(t); \  , \forall \Cc\subset\Bc^{(A)}, |\Cc|\leq L_{\max}^{(A)}, \forall t \text{ in Band-}A \right\}$ is feasible with respect to the UCS based on Defn.~\ref{first-arch}, if it satisfies the following: 
\begin{enumerate}[(i)]
\item For each $t$, the policy assigns RB $t$ to clusters with $\Cc\!\subset\!\Bc^{(A)}$ and $|\Cc|\!=\!\Csize(t)$ in Band-$A$; that is, for each cluster $\Cc$ with $\Uc_\Cc(t)$ being non-empty, we have $\Cc\!\subset\!\Bc^{(A)}$ and $|\Cc|\!=\!\Csize(t)$.
\item For each $t$, each user is served by at most one cluster; that is, $|\sum_{\Cc\subset\Bc}\,\mathbbm{1}\{ k\in  \Uc^{(A)}_\Cc(t)\}|\le 1$.
\item For each $t$ in Band-$A$ and for each  BS $j\in\Bc^{(A)}$, BS  $j$ serves at most $S_{j}(\Csize(t))$ users; that is,
$\left| \cup_{\Cc:  j\in \Cc, \Cc\subset\Bc^{(A)}} \Uc_\Cc(t)\right| \le S_{j}(\Csize(t))$.
\end{enumerate}
\end{defi}

\vspace{-0.4cm}
\subsection{The Feasibility of the NUM Solution in Implementation}
It is easy to verify that $\left\{x_{k\Cc}^{(A)}\right\}$  yielded by any feasible schedules defined by Defn. \ref{feas-sched-one} satisfy (\ref{eq:opt-ct-cluster-cvx})-(\ref{eq:opt-ct-positive-cvx}). 
In fact, when $L_{\max}^{(A)}\!=\!1$ (i.e., cellular cases), there exists at least one feasible schedule that can provide long-term activity fractions approaching the solution of (\ref{eq:opt-cvx}) \cite{BetBur14a}.  
However in the general case $L_{\max}>1$ this is not necessarily true. For instance, for networks with cluster combinations $\{j_1,j_2\}$, $\{j_1,j_3\}$ and $\{j_2,j_3\}$, where $j_1,j_2$ and $j_3$ are BS indexes, 
there exist  $\left\{x_{k\Cc}^{(A)}\right\}$ satisfying   (\ref{eq:opt-ct-cluster-cvx})--(\ref{eq:opt-ct-positive-cvx}), for which no feasible schedule of Defn.~\ref{feas-sched-one} exists.

\begin{theo}\label{theo:feasible}
In the UCSs with $L_{\max}^{(A)}>1$ in some Band-$A$ and with the type of cluster combinations  $\{j_1,j_2\}$, $\{j_1,j_3\}$ and $\{j_2,j_3\}$, where $j_1,j_2$ and $j_3$ are BSs in $\Bc^{(A)}$, there exist some activity fractions satisfying (\ref{eq:opt-ct-cluster-cvx})-(\ref{eq:opt-ct-positive-cvx}) that cannot be implemented by any feasible schedule in Defn.~\ref{feas-sched-one}. 
\end{theo}
\begin{proof}
See Appendix \ref{pf:theo-feasible}.
\end{proof}
\noindent Hence, the coarser time-scale NUM problem (\ref{eq:opt-cvx}) does not capture the finer time-scale  constraints associated with feasible schedulers. Although, in general, (\ref{eq:opt-cvx}) provides an upper bound on the network performance, as we show next, using activity fractions that are the solution to (\ref{eq:opt-cvx}), we can design scheduling policies, whose performance is close to the utility provided by the solution to~(\ref{eq:opt-cvx}).

\vspace{-0.4cm}
\subsection{Virtual Queue Based Scheduling Scheme}\label{subsec:vq-scheme}
\vspace{-0.2cm}
We next present scheduling policies for the UCS architecture  comprised of $\sum_{A=1}^3\Csize_{\max}(A)$ parallel schedulers, one per each subband. 
We describe a method for scheduling users over the RBs from the $\lambda_{A\Csize}>0$ fraction of RBs dedicated to clusters of size $\Csize$ in band-$A$.

Given the limited number of fractional users per cluster size $\Csize$, the scheduler approximates the optimal  $\{x_{k\Cc}^{(A)}\}$ by unique association activity fractions,   $\{\tilde{x}_{k\Cc}^{(A)}\}$, given by
\begin{equation}
\label{xtil-kC}
 \tilde{x}_{k\Cc}^{(A)}=
\begin{cases} 
 x_{k\Cc}^{(A)} & \text{if  $\Cc = \Cc^*(k) $}\\
 0 &  \text{otherwise} \end{cases},
\end{equation}
with $
\Cc^*(k) = \argmax_{\Cc: \ |\Cc|=L, \Cc\subset\Bc^{(A)}}  x_{k\Cc}^{(A)}$.
Letting $\mathcal{U}_\Cc^{(A)}$ denote the users for which $\tilde{x}_{k\Cc}^{(A)} > 0$, we  have  $\mathcal{U}_\Cc^{(A)}\cap\mathcal{U}_{\Cc'}^{(A)}\!=\!\emptyset$  for all   $\Cc\!\neq\!\Cc'$ with $|\Cc|\!=\!|\Cc'|$. We also let $\mathcal{U}^{(A\Csize)}\!=\!\cup_{\Cc:\, |\Cc|\!=\!\Csize} \, \mathcal{U}_{\Cc}^{(A)}$ denote the set of users that receive non-zero activity fractions from clusters of size $\Csize$ in Band-$A$. In the rest of this section, we focus on clusters $\Cc$ satisifying $|\Cc|\!=\!\Csize$ and $\Cc\!\in\!\Bc^{(A)}$, unless otherwise specified. 

To assign user $k$ a fraction of RBs close to the desired fraction in the $L^{\rm th}$ subband of Band-$A$, i.e.,  $\alpha_{k}\!=\!\tilde{x}_{k\Cc^*(k)}^{(A)}/\lambda_{A\Csize}$, we consider a max-min scheduling  policy based on virtual queues (VQ), which assumes user $k$ receives rate $\tilde{R}_{k}\!=\!1/\alpha_{k}$ when user $k$ is scheduled for transmission by cluster $\Cc^*(k)$ (i.e., $k\in \Uc^{(A)}_{\Cc^*(k)}(t)$). 
The  cluster-size $\Csize$  scheduler performs at each $t$ a weighted sum rate maximization (WSRM) of the form \cite{ShiCai10}:
\begin{subequations}\label{VQ-C}
\begin{align}
\max_{\tilde{\Uc} \subseteq \Uc^{(A\Csize)}}\ & \ \sum_{ k \in \tilde{\Uc}  } Q_k(t)\tilde{R}_{k},
\label{VQ-C-WSRM}
\\
\text{s.t. } \ \ \  &  \  \   \sum_{k\in \tilde{\Uc}} \mathbbm{1}\{j\in \Cc^*(k)\} \le S_{j}(\Csize), \  \ \forall j\in \Bc,
\label{VQ-C-constraints}
\end{align}
where the weight of user $k$ at time $t$, $Q_k(t)$, is  the VQ length of user $k$ at time $t$. For max-min fairness \cite{ShiCai10},   $Q_k(t)$ is updated by
$
Q_k(t+1) = \max\{0, Q_k(t) - \tilde{R}_k(t)\} + A_k(t), 
$ 
where
\begin{equation}
\tilde{R}_k(t) = \begin{cases} 
 \tilde{R}_k  & \text{if user $k$ is scheduled at time $t$}\\
 0 &  \text{otherwise} \end{cases},
 \end{equation}
\begin{equation}
A_k(t)  = \begin{cases} 
 A_{\rm max}  & \text{if $V > \sum_k Q_k(t)$} \\
 0 &  \text{otherwise} \end{cases},
 \end{equation}
\end{subequations}
with $A_{\max}$ and $V$ chosen sufficiently large \cite{ShiCai10}. 
Note that in the absence of  constraints (\ref{VQ-C-constraints}), the max-min scheduler  (\ref{VQ-C})  schedules user $k$ the desired fraction of RBs, $\alpha_k$.

Scheduling via (\ref{VQ-C}) is impractical, as it amounts to solving for each RB $t$ an integer linear program~(\ref{VQ-C}). A number of heuristic algorithms can be used to provide feasible (though generally suboptimal) solutions to (\ref{VQ-C}).  In this paper, we consider a rudimentary greedy algorithm. Letting $K_{A\Csize}\!=\!|\Uc^{(A\Csize)}|$ be the total number of users to be served by clusters of size $\Csize$, the greedy algorithm for size-$L$ clusters at time $t$ operates as follows:
\begin{enumerate}[1.]
\item Determine a user order $\pi(k)$, where  $Q_{\pi(k)}(t) \tilde{R}_{\pi(k)} \ge Q_{\pi(k+1)}(t) \tilde{R}_{\pi(k+1)}$  for all $k\in\Uc^{(A\Csize)}$.
\item Initialization: $k=1$, and $\tilde{\Uc}=\emptyset$.
\item  If the user set $\tilde{\Uc} \cup \{\pi(k)\}$ satisfies all the constraints in (\ref{VQ-C-constraints}), set $\tilde{\Uc} =\tilde{\Uc} \cup \{\pi(k)\}$.
\item If $k< K_{A\Csize}$,  set $k=k+1$ and go to step 3.
\item Output $\tilde{\Uc}$ as the scheduling user set for size-$L$ clusters in Band-$A$ at time $t$.
\end{enumerate}

\section{Performance Evaluation}\label{sec:simulation}
In this section, we present a simulation-based evaluation based on the ``wrap-around'' layout in Fig.~\ref{fig:network}. We also present the simulation results with the network deployment including more hexagonal modeled macrocells and non-uniformly distributed users (based on 3GPP layout in TR 36.872).\footnote{We assume full-buffer traffic model, while the study of more general traffic models are left for future work.}

The parameters used for the layout in Fig.~\ref{fig:network} are given as follows unless otherwise specified. 
There are 4 macros with $M_j\!=\!100$ and  $S_j(|\Cc|)\!=\!\max\{10\rho|C|, 10\}$,  and 32 small cell BSs  with  $M_j\!=\!40$ and $S_j(|\Cc|)\!=\!\max\{4\rho|C|, 4\}$, where $\rho$ is a tunable parameter in $[0,1]$. There is 1 small cell BS at the center of each white square, while 3 small cell BSs being dropped uniformly within each shaded square (hotspot).  Also, 15 and 90  single-antenna users are dropped uniformly in each white and shaded square, respectively. 
The macro and small cell BS transmit powers are 46dBm and 35dBm, respectively. The path-loss for  macro-user links and small cell BS-user links are $128.1\!+\!37.6\log_{10}d$ and $140.7\!+\!36.7\log_{10}d$,  respectively, with the distance $d$ in km.  The noise power spectral density is $-174$ dBm/Hz. 

We consider three distinct macro-small cell resource sharing scenarios: (i) the shared scenario with macros and small cell BSs transmitting on the same RBs -- operations with $A\!=\!1$; (ii) the orthogonal scenario with macros and small cell BSs transmitting on different bands -- operations with $A\!\in\!\{2,3\}$, where we provide macros (Band-$2$) 20\% RBs as an illustrative example; (iii) RB blanking with macros muted on certain RBs -- operations with $A\!\in\!\{1,3\}$. 
Note that, although we can jointly optimize the resource partition (i.e., $\mu_A$ ) among different bands and user activity fractions (i.e., $x_{k\Cc}^{(A)}$) using (\ref{eq:opt-cvx}) in the orthogonal scenario, we fix $\mu_A$ due to the following reasons. The resource partition among macro and small cells is most likely static (or semi-static) in practice. Moreover, the macro and small cells may operate on different frequency bands (e.g., the macro and small cells may transmit on lower-frequency bands and higher-frequency bands, respectively), where $\mu_A$ then depends on the available resources on each band and thus is not a variable to optimize. As for the selection of the fixed values for $\mu_A$, we set $\mu_2 = 0.2$ and $\mu_3 = 0.8$ as an illustrative example. In our selection, we let $\mu_3$ to be larger than $\mu_2$ since the small cells are deployed more densely than the macro BSs. Our simulation results can be easily updated with other values. For completeness, we have provided the simulation results in the orthogonal scenarios with different $\mu$ at the end of subsection \ref{sec:simu-layout1}.

We make comparisons between the conventional approach (i.e., max-SINR), the approach of \cite{BetBur14a} and the proposed UCS of this work. Both max-SINR and \cite{BetBur14a} are cellular approaches. In fact, the formulation of \cite{BetBur14a} is equivalent to UCS in Scenario (i) with $L_{\max}(1) = 1$. 

$L_{\max}^{(A)}$ depends on the band: For our simulations, we consider $L_{\max}^{(1)}\in\{1,4\}$ and $L_{\max}^{(3)}\in\{1,4\}$. For Band-$2$, only cellular transmission from macro BSs are allowed, and hence $L_{\max}^{(2)}\!=\!1$. The number of all possible  clusters of size greater than $4$ is too large for any practical purpose. Besides, not all subsets of BSs are good candidates for being clusters. We determine the set of potential BSs from the perspective of users: we let each user pick the strongest 8 BSs providing the largest signal strength to that user\footnote{We pick the strongest 8 BSs, since the performance of picking the strongest 9 BSs is almost the same as the 8-BS case, while the utility of picking the strongest 7 BSs is less than the 8-BS case.}, and the potential BS clusters that can serve the user only include BSs among these 8 BSs.

\begin{figure}
\centering
\includegraphics[width=7cm, height=6.5cm]{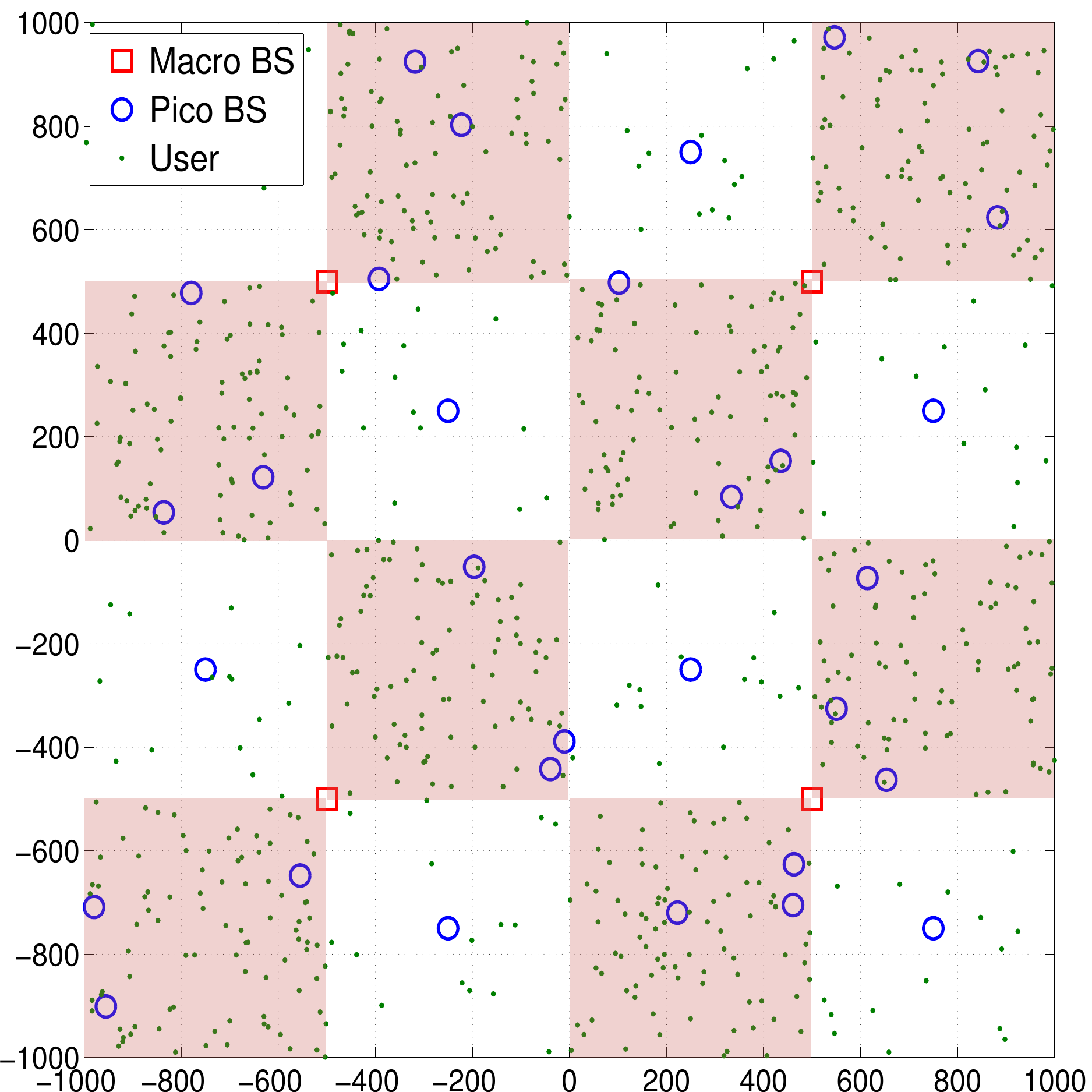}
\caption{The illustration of network deployment. The white grids are the regular areas, while the shadowed grids are hotspots.}
\label{fig:network}
\end{figure}

There is a one-to-one mapping between the log utility and the geometric mean of rates as $\left(\prod_{k=1}^{K} R_{k}\right)^{\frac{1}{K}}$ $=\!\exp\left(\frac{1}{K}\sum_{k=1}^K \log R_{k}\right)$, thus we use geometric mean of rates as the metric for performance evaluation.

\subsection{Simulation of Layout 1 (Figure \ref{fig:network})}\label{sec:simu-layout1}
Fig. \ref{fig:geomean-rate} show the geometric mean of rates in scenarios (i) and (ii). The optimal solution to  (5),  hereby denoted as UCS-NUM, is obtained by CVX. 
We provide performance comparisons between the CVX solution (denoted as the UCS-NUM) of (5) and the solution of the dual subgradient based algorithm. The latter has almost the same performance as the NUM solution, which validates our analysis. We observe this also for our later simulations, hence the results of the dual algorithm are skipped in following figures for the sake of clarity. 
It can be seen that when the solution to (\ref{eq:opt-cvx}) is approximated by (\ref{xtil-kC}) with unique association, the utility loss is insignificant thanks to a very few number of fractional users (as shown in Prop. \ref{prop:uniqueass}). Moreover, the proposed greedy VQ  scheduling scheme provides performance close to the NUM solution, and in particular within 90\% of the utility provided by the NUM solution in both scenarios (i) and (ii). Note that in  cellular transmission, the NUM solution is feasible via some scheduler, and thus VQ based scheduling is unnecessary \cite{BetBur14a}. 
We can observe that the UCS significantly improves the geometric mean of rates versus the optimal cellular performance and the max-SINR  association (about 1.6$\times$  in the shared scenario and 1.35$\times$  in the orthogonal scenario versus the  optimal cellular result). 

In Fig. \ref{fig:geomean-rate-ic}, we compare the performance of scenarios (i) and (iii). We can observe that RB blanking further improves the network utility.

\begin{figure}
\centering
		\subfigure[Scenario(i) and Scenario(ii)]{
			\label{fig:geomean-rate}
			\includegraphics[width=7.3cm, height=6.8cm]{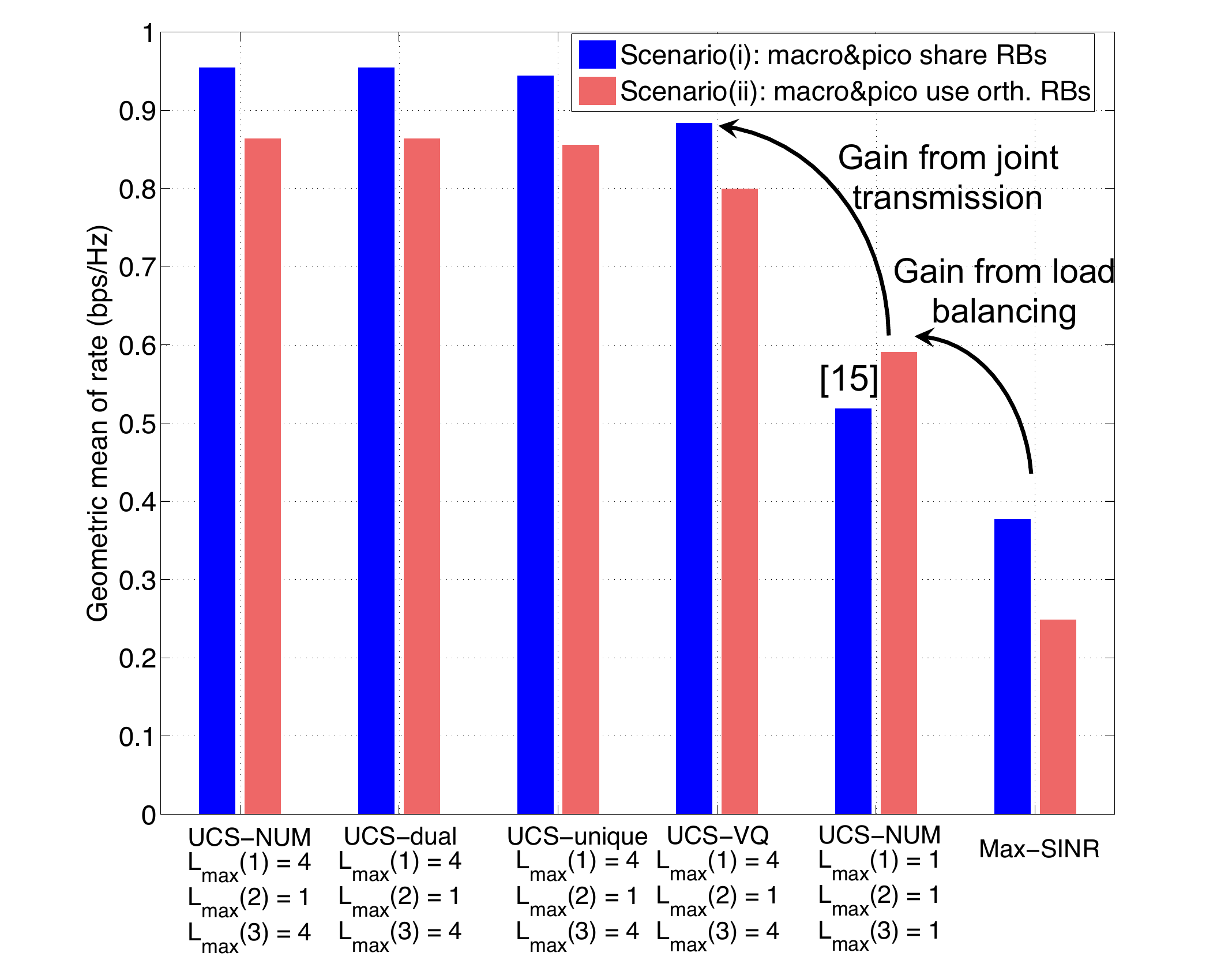}}
		\subfigure[Scenario(i) and Scenario(iii)]{
			\label{fig:geomean-rate-ic}
			\includegraphics[width=7.5cm, height=6.8cm]{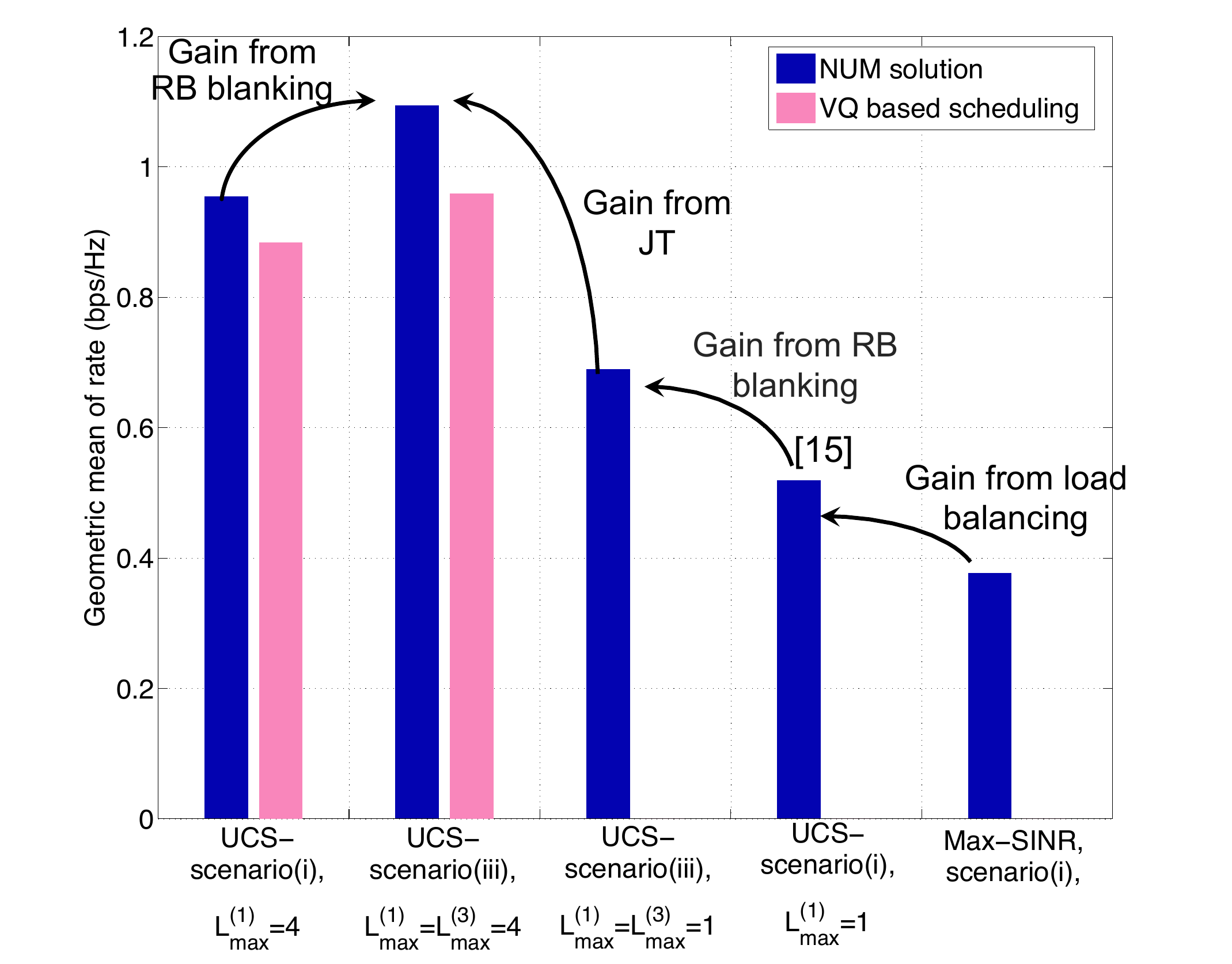}}
\caption{The geometric mean of rates using different approaches ($\rho\!=\!1$): (a) The UCS with VQ based scheduling scheme provides a large performance gain (about 1.6$\times$ and 1.35$\times$ in the shared and  orthogonal scenarios, respectively) versus the optimal cellular result. (b) RB blanking further improves the network performance. }
\label{fig:geomean-rate-all}
\end{figure}

Observation of Fig. \ref{fig:Rcdf} yields similar conclusions.  Indeed Fig. \ref{fig:Rcdf} shows the rate cumulative distribution function (CDF) with different approaches. We illustrate the results of the shared and RB blanking scenarios in the same figure, as RB blanking is essentially motivated from the shared scenario to manage the interference from macros to small cell users. 
The rate of bottom (the 10th percentile)  users using UCS in scenario (i) is about 2.2$\times$ of the optimal cellular solution of \cite{BetBur14a}. The gain is even larger in scenario (ii).

\begin{figure}
\centering
		\subfigure[Scenarios (i) and (iii)]{
			\label{fig:Rcdf-share}
			\includegraphics[width=7.8cm, height=6.5cm]{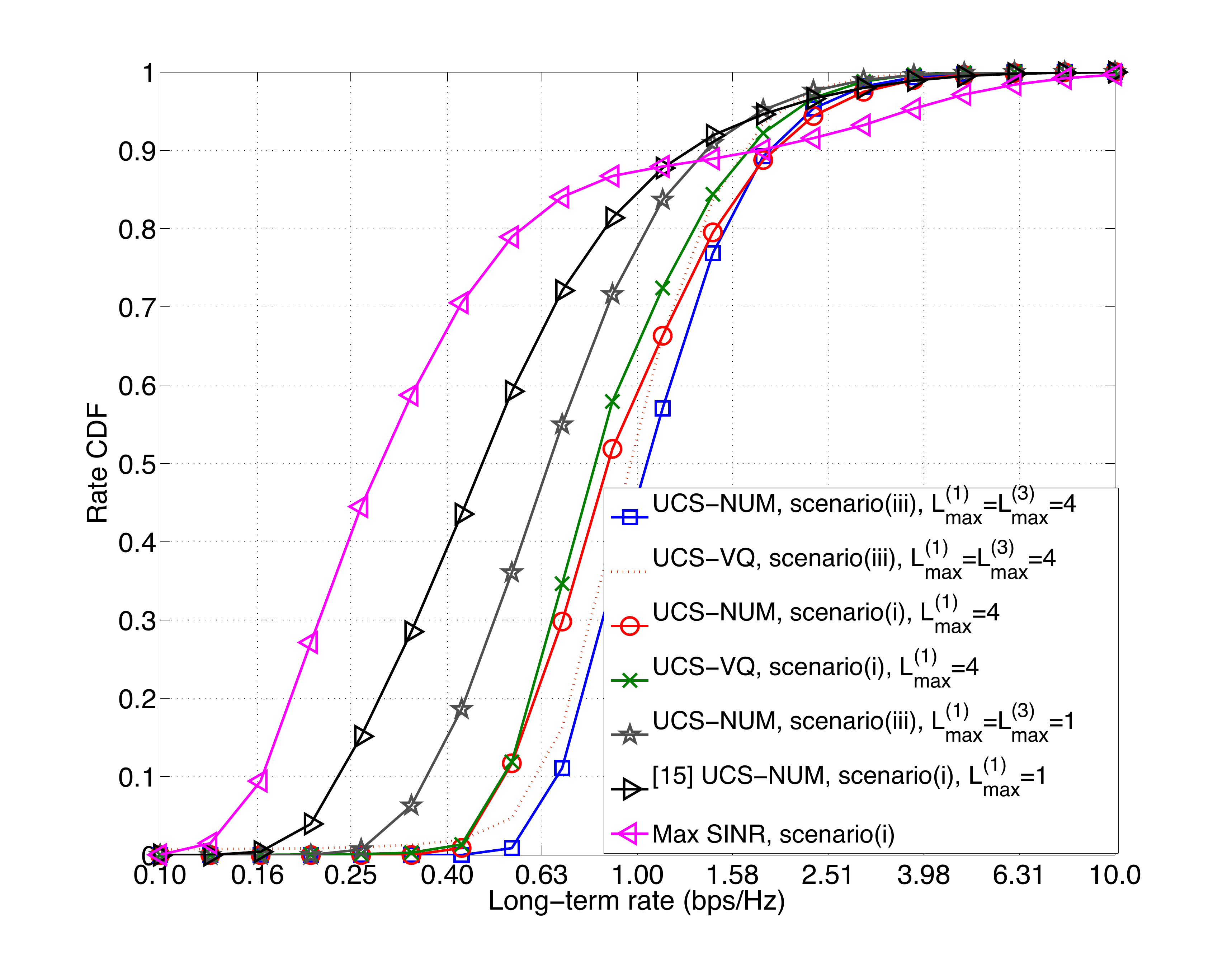}}
		\subfigure[Scenario (ii)]{
			\label{fig:Rcdf-orth}
			\includegraphics[width=7.8cm, height=6.5cm]{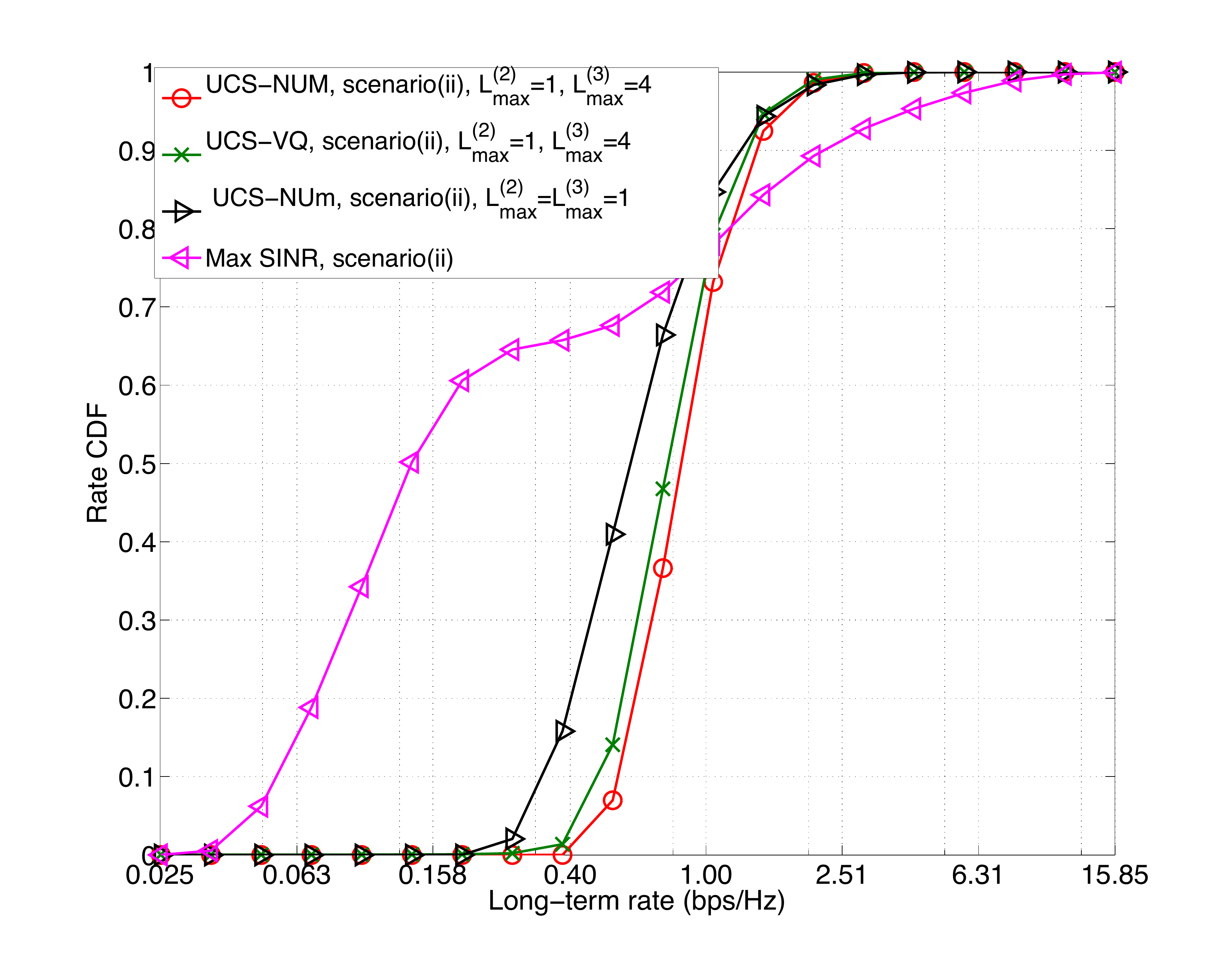}}
\caption{The long-term rate CDF using different approaches ($\rho\!=\!1$). The rate of bottom (10th percentile)  users using UCS is about 2.2$\times$ of the cellular transmission case with optimal user association but without interference management. 
}
\label{fig:Rcdf}
\end{figure}

The number of users served by different clusters with UCS is illustrated in Fig. \ref{fig:UEAss}. In the shared scenario with max-SINR association, most users connect to macro BSs, since macro BSs have much larger transmit power than small cell BSs. By load balancing, many users are offloaded to small BSs in the optimal cellular solution. In our proposed framework, all users are served by BS clusters with multiple BSs, which implies the potential gain using UCS. In the orthogonal scenario, there is no cross-tier interference and more users may get larger SINR from small BSs than macro BSs, hence more users connect to small cell BSs in the max-SINR association compared to the shared scenario.  
Due to the limited resources  (20\% RBs) available in macro BSs, more users are offloaded to small BSs using the load balancing approach in orthogonal cellular transmission.  For scenario (i), the percentage of fractional users  is about 3.3\% using UCS, and 1.2\% in the case with optimal cellular. 
In the RB blanking scenario, the percentage of fractional users in the case using LJT with  RB blanking (scenario (iii)) is about 2.5\%, while the percentage of fractional users adopting cellular transmission with blanking is less than 1\%. Thus, we can conclude that the number of fractional users in all cases is very small, which validates our analysis. 

\begin{figure}
\centering
	\subfigure[Scenario (i)]{
			\label{fig:UEAss-share}
			\includegraphics[height=6cm]{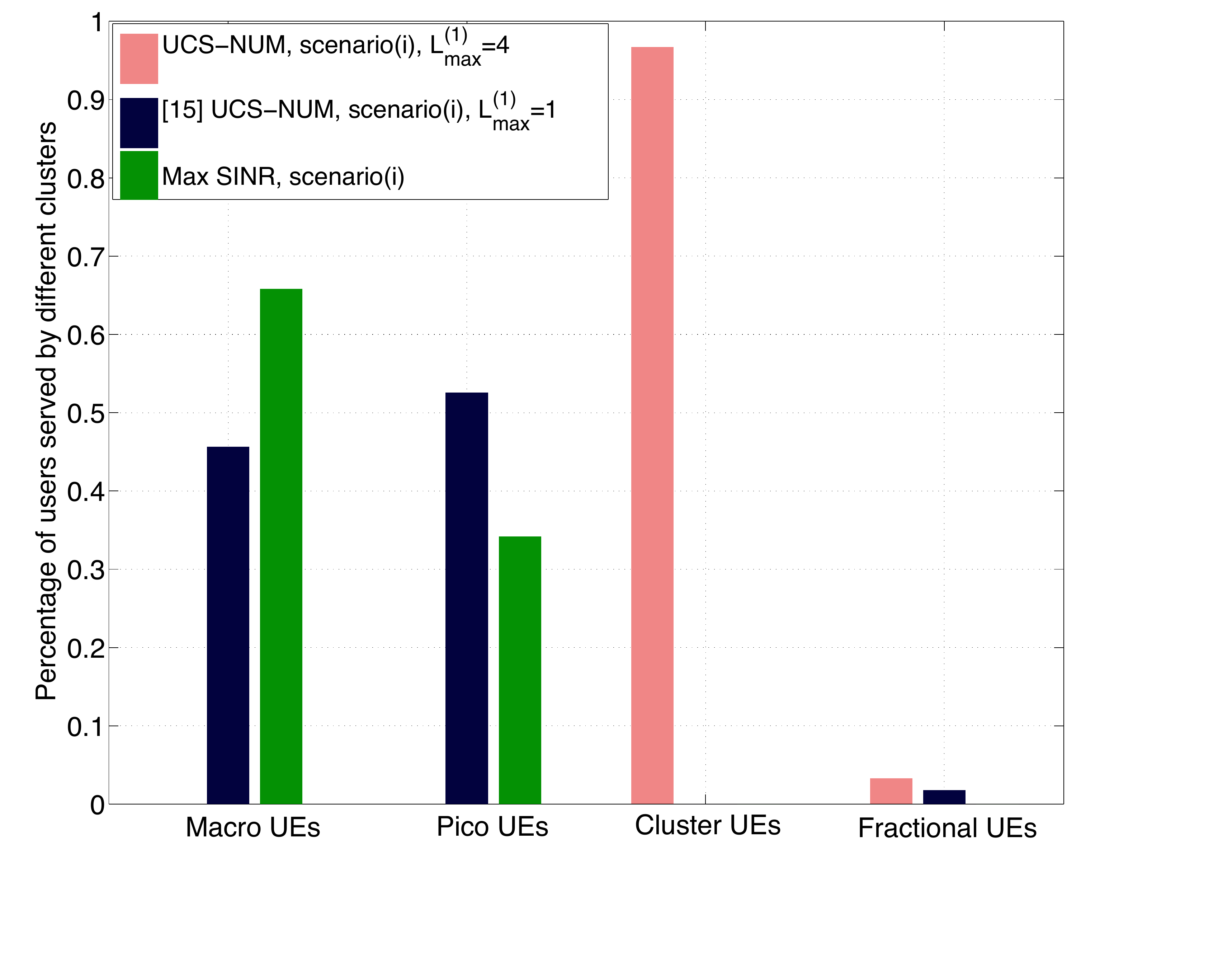}}
		\subfigure[Scenario (ii)]{
			\label{fig:UEAss-orth}
			\includegraphics[height=6cm]{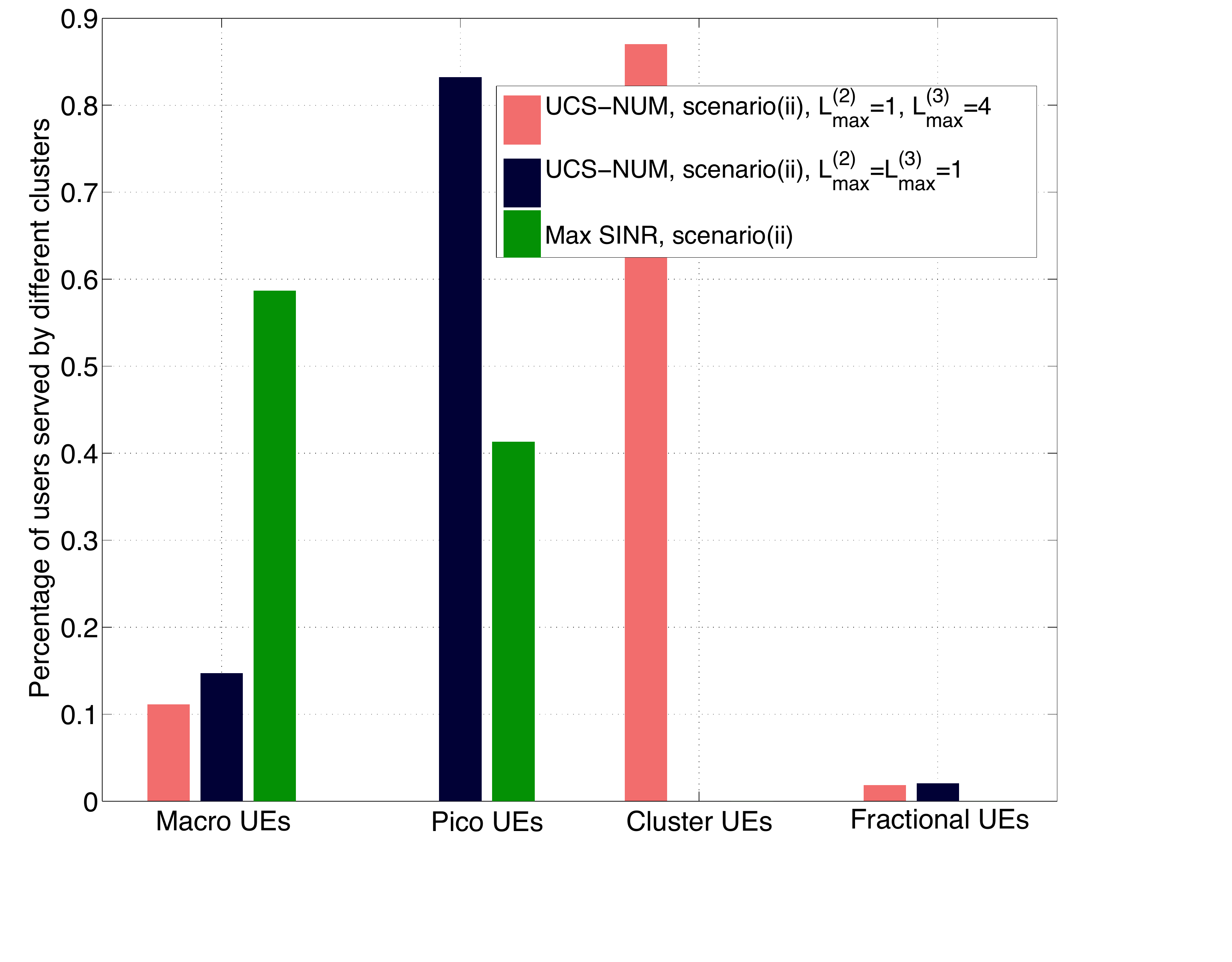}}
\caption{The number of users served by different clusters ($\rho\!=\!1$). Most users have unique association. The ``Cluster UEs'' refer to the users served by  clusters of size larger than 1. 
 }
\label{fig:UEAss}
\end{figure}

\begin{figure}
\centering
\subfigure[]{\label{fig:geomean-rate-rho}
\includegraphics[width=8.4cm, height=7.4cm]{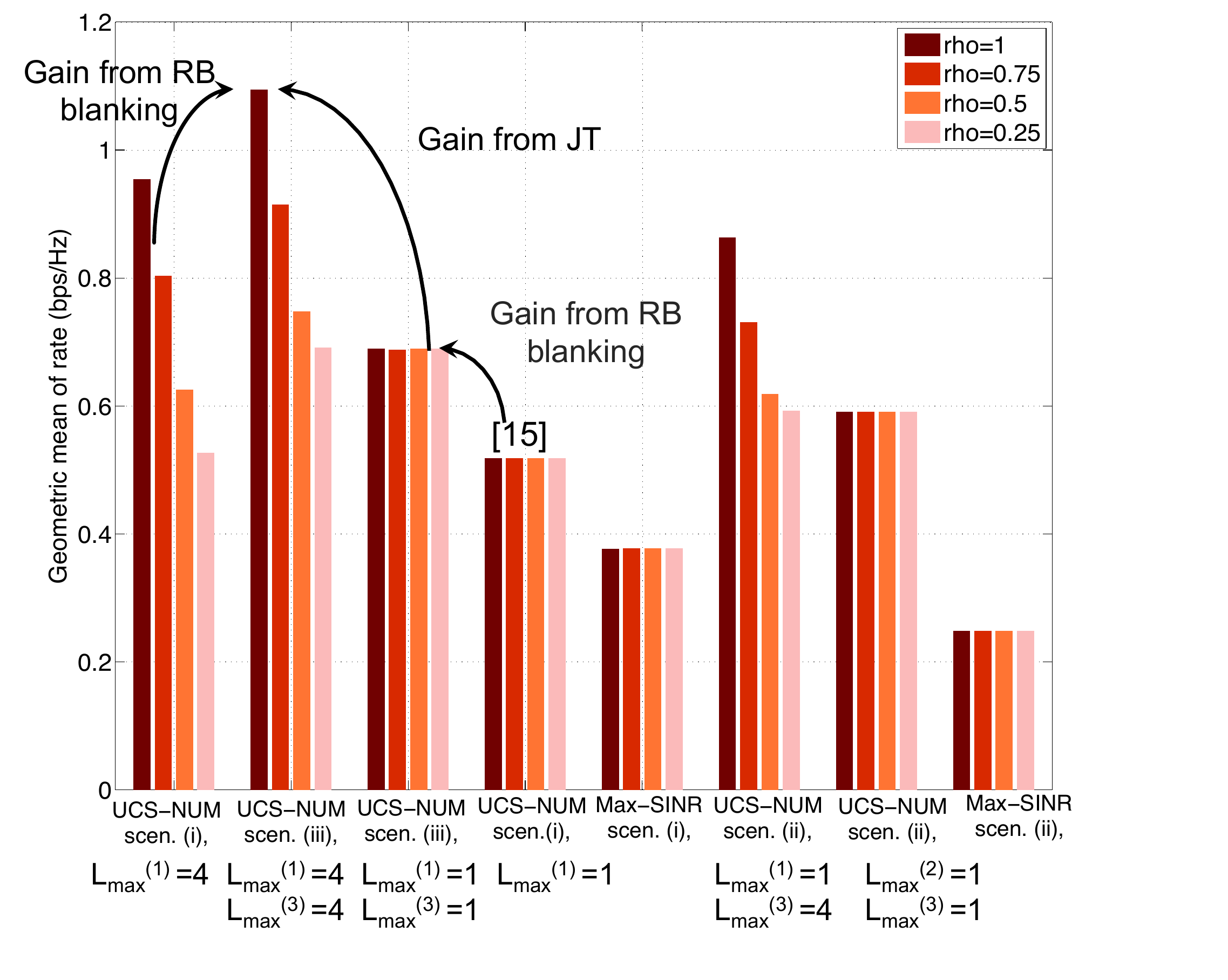}}
\subfigure[]{\label{fig:BStime-rho}
\includegraphics[width=7.8cm, height=6.9cm]{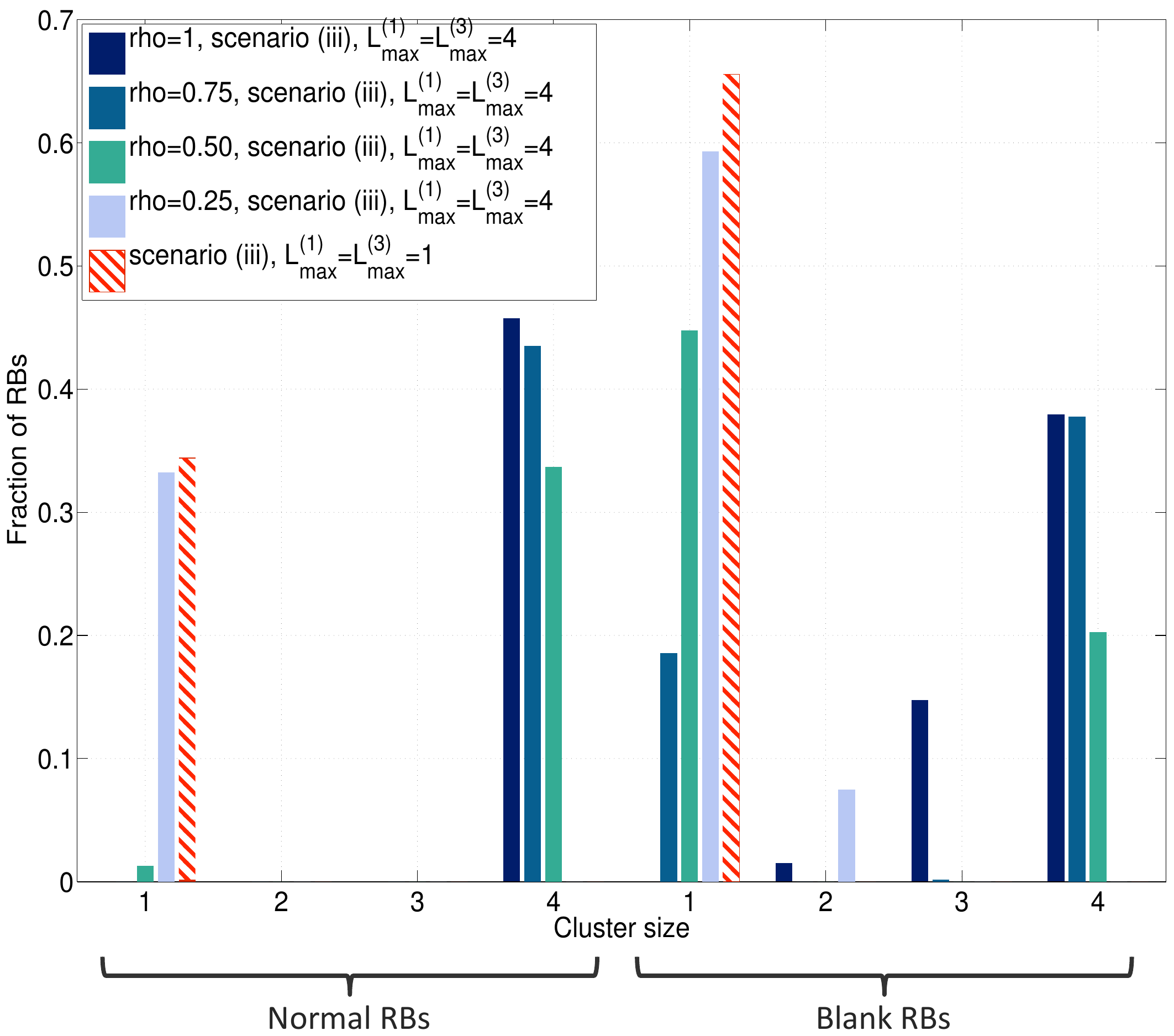}
}
\caption{(a) The geometric mean of rates using different approaches versus $\rho$. As $\rho$ decreases, the gain from JT decreases in both shared and orthogonal scenarios. (b)The fraction of resources allocated to clusters of different sizes in the RB blanking scenario. As $\rho$ decreases, more resources are allocated to clusters with smaller size.
}

\end{figure}

In Fig. \ref{fig:geomean-rate-rho}, we show the geometric mean of rates versus different $\rho$. We observe that the performance gain using UCS decreases as $\rho$ decreases, since the number of users that can be served by clusters decreases.   
This implies that the gain from UCS increases as more UL pilot resources are available in the system. With limited UL pilot resources, the gain from UCS would be quite small.

%

Fig. \ref{fig:BStime-rho} illustrates the resource allocation for clusters of different sizes versus $\rho$ in the RB blanking scenario. The macro BSs are off for about 65\% RBs in cellular transmission. In scenario (iii),  as $\rho$ decreases, the clusters serve less users, and more resources are allocated to the clusters of smaller sizes. When $\rho=0.25$, all resources are allocated to single-BS clusters in normal RBs, and most of the resources are allocated to single-BS clusters in blank RBs. This again suggests that when the available pilot resources are strictly constrained, the gain from  LJT would be limited. 

Fig. \ref{fig:orth-mu} illustrates the simulation results in the orthogonal scenario (Scenario (ii)) with different values of $\mu$. 
As $\mu_2$ increases, the utility  of max-SINR increases due to the fact that most users are associated to macro BSs in max-SINR association. As $\mu_2$ increases, more  resources are available to macro users and thus the utility can be improved. On the other hand, if the available resources are limited for small cells, it would result in less motivation for load balancing (i.e. pushing users off from macro to small cells), since the users, if offloaded to small cell, still suffer limited resources and thus limited rate. Therefore, the gain from load balancing and JT will be limited when small fraction of resources are allocated to small cells (i.e. when$\mu_2$ is large), as can be observed from Figure \ref{fig:orth-mu}. 

\begin{figure}
\centering
\includegraphics[width=9.0cm, height=8cm]{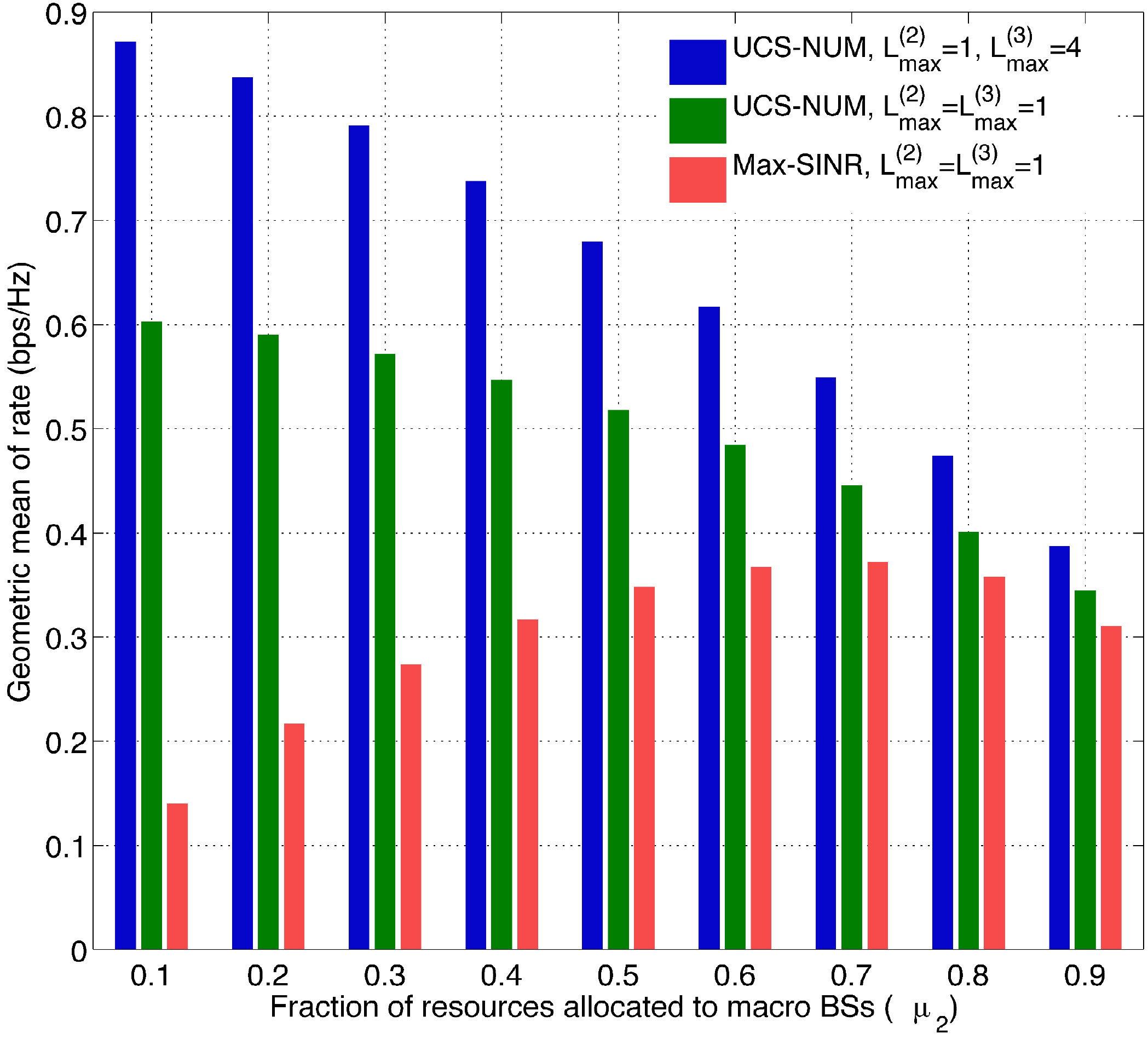}
\caption{The illustration of geometric mean of rates in Scenario (ii) (the orthogonal scenario) with different fractions of resources allocated to macro and small cell layers.}
\label{fig:orth-mu}
\end{figure}

\subsection{Simulation of Layout 2 (Figure \ref{fig:network-3gpp})}
In this subsection, we provide similar simulation results for a network topology complaint with 3GPP HetNet scenario \cite{3GPP872} as shown in Fig. \ref{fig:network-3gpp}. In particular, we have a cellular layout with 7 macro-cell BSs and 3 hotspots per macrocell. Within each hotspot region there are 4 randomly dropped small cell BSs. 120 UEs are uniformly dropped in each hotspot region while 60 more UEs are dropped randomly in the whole coverage area of each macro cell. The macro/small cell powers and the pathloss models used in this experiment are identical to those used in the previous layout. 

Fig. \ref{fig:geomean-rate-ic-hex-w} compares the geometric mean of rate with various methods in scenarios (i) and (ii), and Fig. \ref{fig:geomean-rate-hex-w} presents the geometric mean of rate with various methods in scenarios (i) and (iii).   Similar to the layout illustrated Fig. \ref{fig:network}, we also observe a significant gain in geometric mean of rate by using LJT in all considered scenarios. Specifically, the UCS with VQ based scheduling scheme provides a large performance gain (about 1.35×) versus the optimal cellular result, as illustrated in both Figs. \ref{fig:geomean-rate-ic-hex-w} and \ref{fig:geomean-rate-hex-w}.

Figs. \ref{fig:Rate-cdf-hex} and \ref{fig:Rate-cdf-orth-hex} illustrate the user rate CDFs with different approaches. An increase of 83\% can be observed for the cell-edge users at the 10th percentile compared to the cellular case with optimal load balancing but no interference management. It can also be observed that joint RB blanking and  JT further improve the network performance.

\begin{figure}
\centering
\includegraphics[width=8.6cm, height=7.2cm]{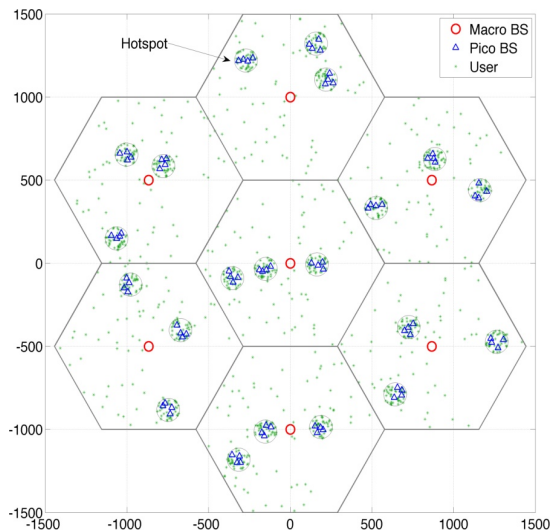}
\caption{The illustration of 3GPP layout.}
\label{fig:network-3gpp}
\end{figure}

\begin{figure}
\centering
\subfigure[Scenarios (i) and (ii)]{\label{fig:geomean-rate-ic-hex-w}
\includegraphics[width=7.8cm, height=6.9cm]{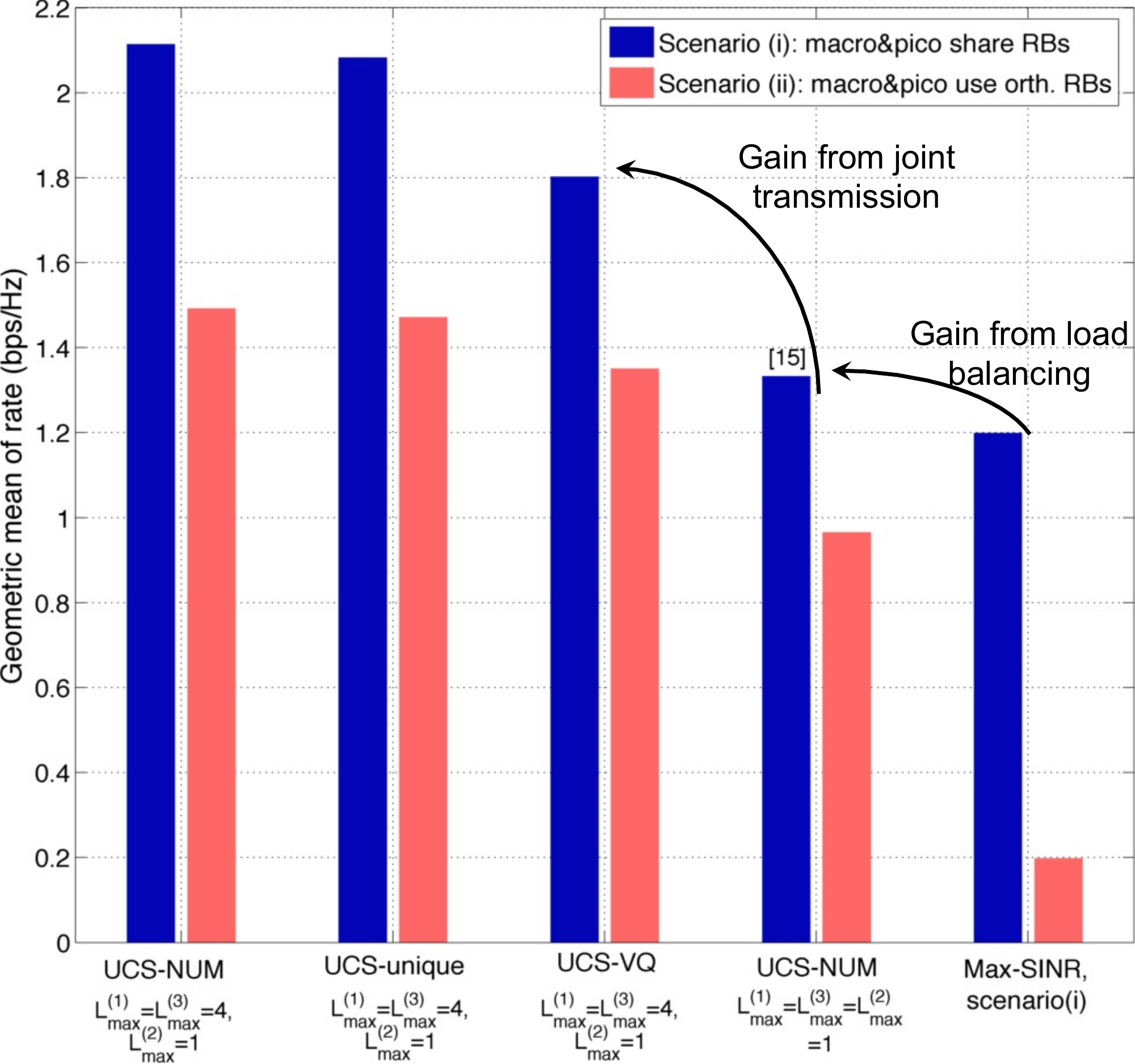}}
\subfigure[Scenarios (i) and (iii)]{\label{fig:geomean-rate-hex-w}
\includegraphics[width=7.8cm, height=6.9cm]{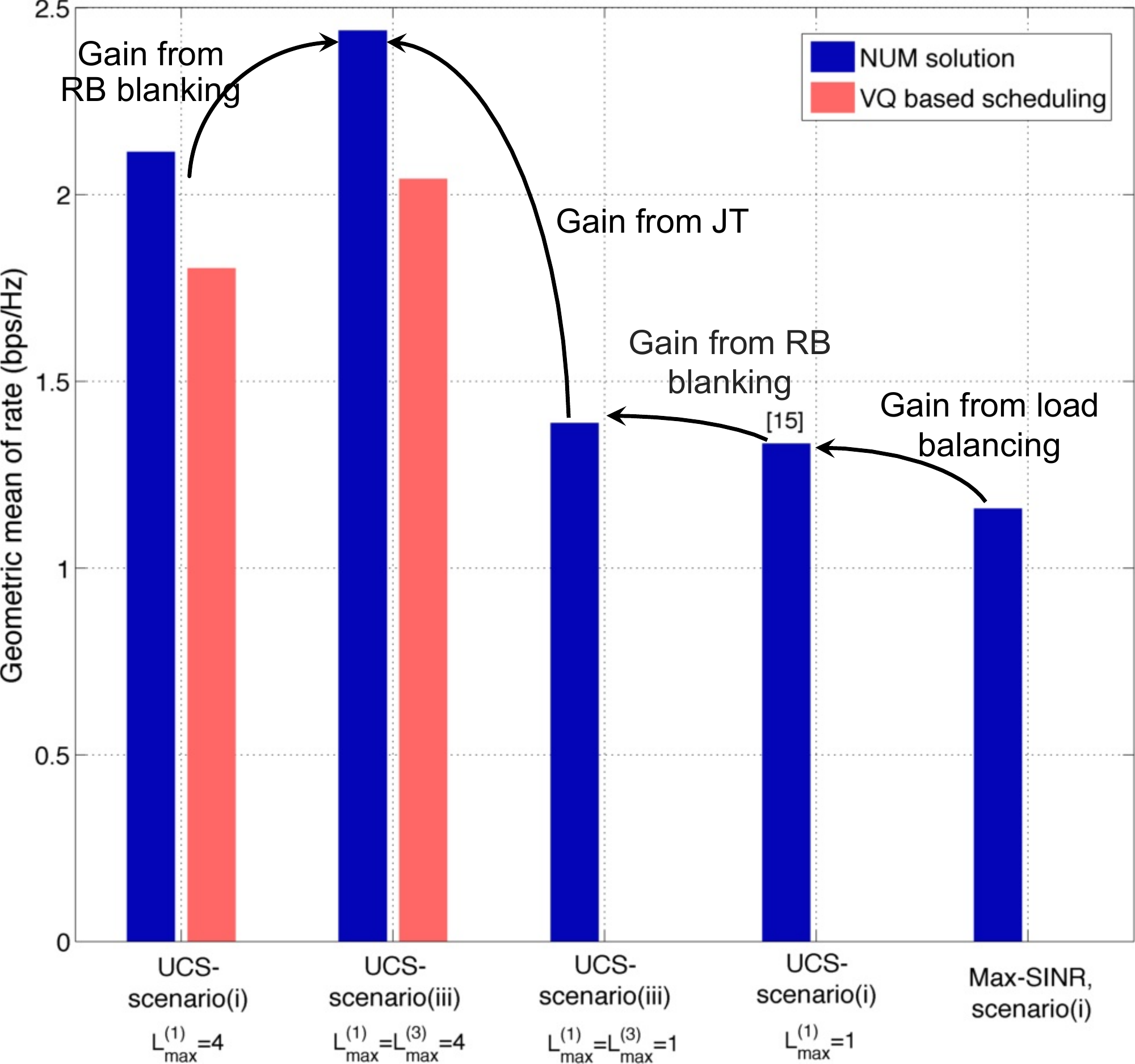}}
\caption{Geometric mean of rate in 3GPP layout with $\rho = 1$.}
\end{figure}

\begin{figure}
\centering
\subfigure[Scenarios (i) and (iii)]{\label{fig:Rate-cdf-hex}
\includegraphics[width=7.8cm, height=6.9cm]{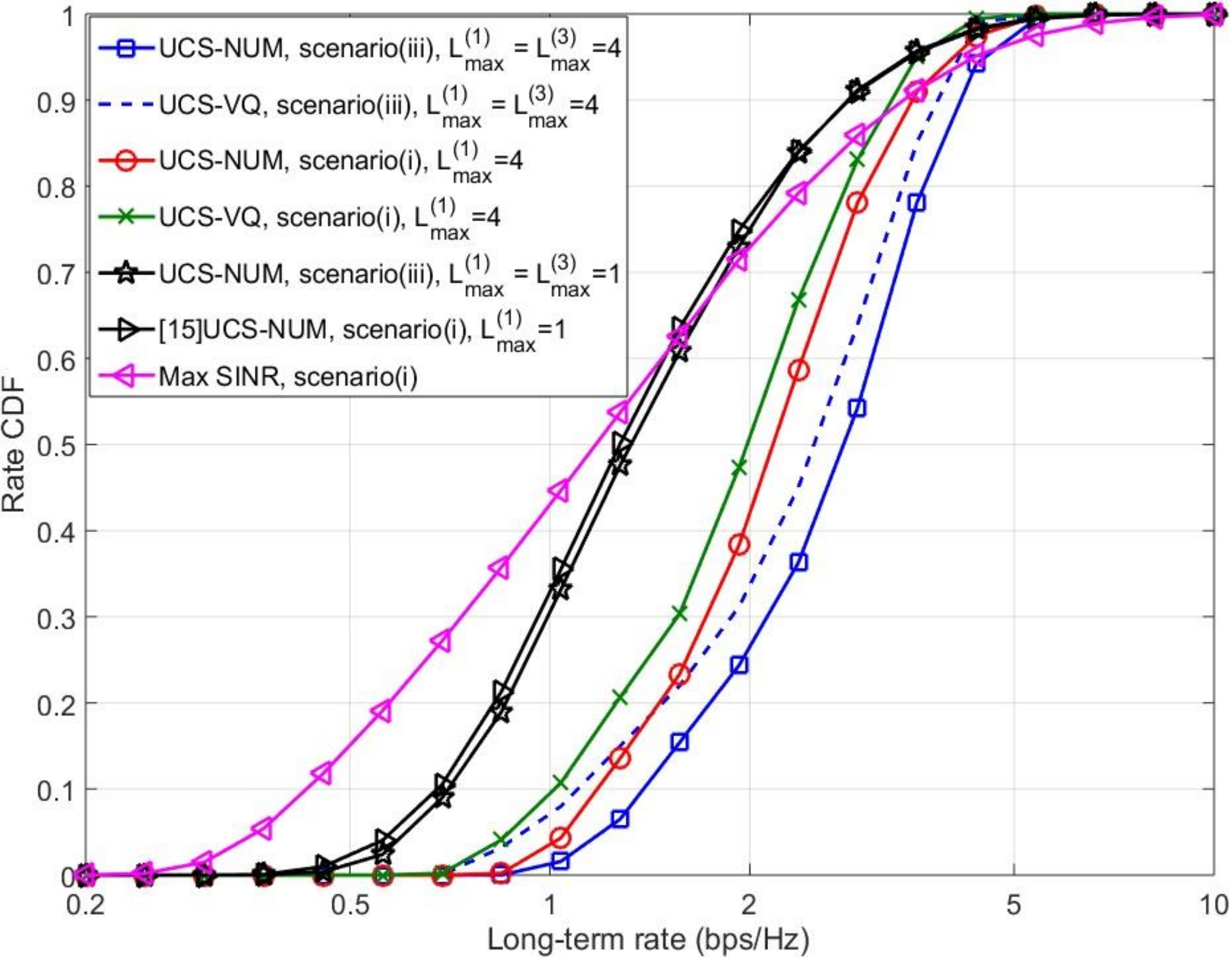}}
\subfigure[Scenario (ii)]{\label{fig:Rate-cdf-orth-hex}
\includegraphics[width=7.8cm, height=6.9cm]{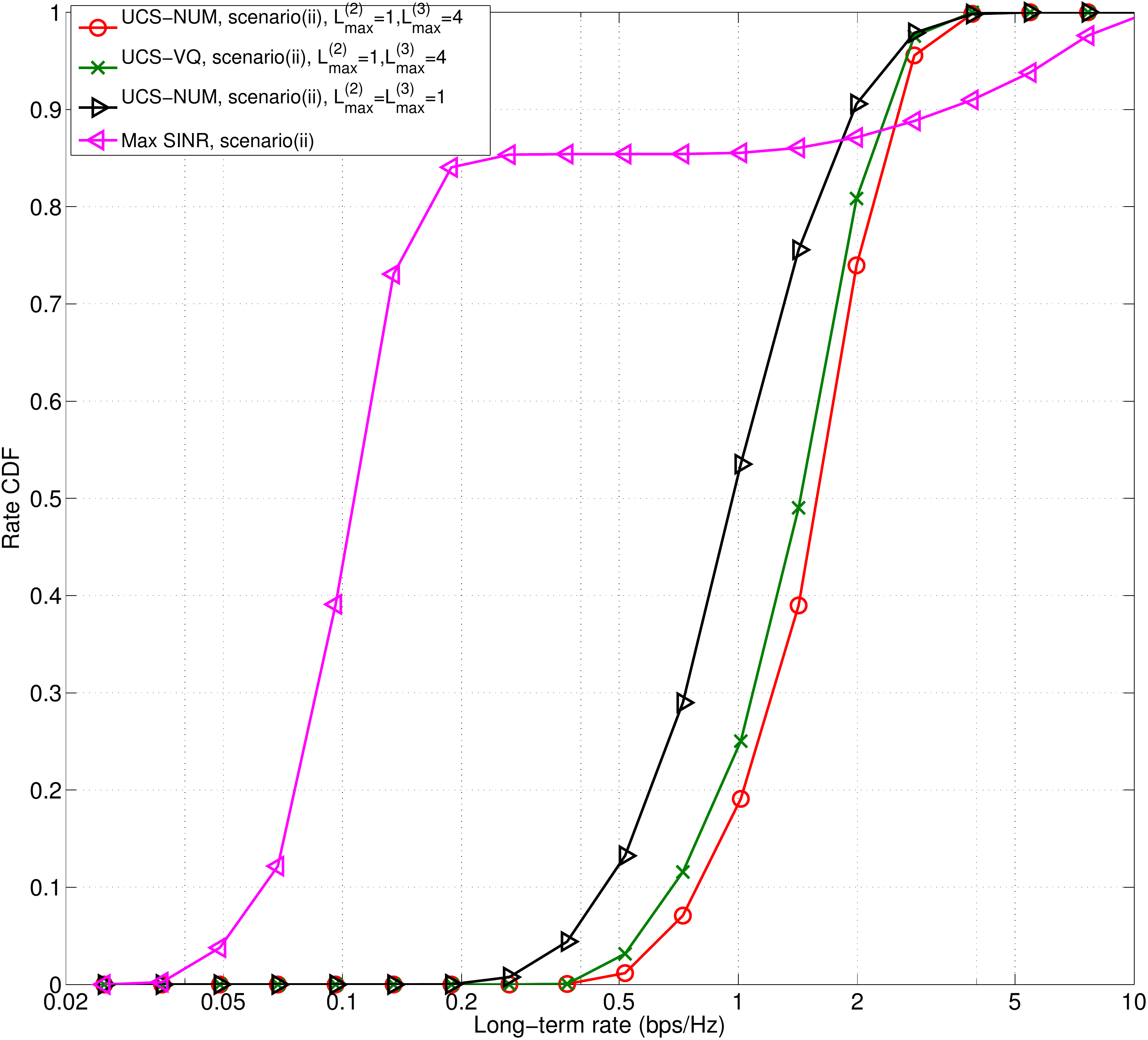}}
\caption{The long-term rate CDF in 3GPP layout with $\rho = 1$.}
\end{figure}

%
%


\section{Conclusion}\label{sec:conclusion}
In this paper, we investigate the joint optimization problem of user association and interference management in the massive MIMO HetNets. We  consider both LJT and RB blanking approaches for interference management. 
We first provide the instantaneous rate from BS clusters by exploiting  massive MIMO properties, namely the rate hardening and the independence of peak rate from the user scheduling. 
We then formulate a convex NUM problem to obtain the optimal user-specific BS clusters and the corresponding resource allocation. The unified formulation can be applied to both LJT and blanking approaches, as well as the case where macro and small BSs use orthogonal resources. We further propose an efficient dual subgradient based algorithm, which is shown to converge towards the NUM solution. 
We show that the NUM solution with LJT may not be implementable by a feasible scheduler, and thus it provides an upper bound on the performance. Showing that most users connect to at most one cluster per RB in heavily loaded networks, we propose to approximate the NUM solution to a unique association, given which we propose a VQ based scheduling scheme to provide approximate but implementable results. Simulations show that the proposed scheduling scheme yields results that closely match the NUM solution. Investigations involving more dynamic settings (e.g., users with high mobility) and the impact of different factors (such as imperfect CSI acquisition,  number of users simultaneously served by BSs in different cluster sizes) on the overall system performance are all subjects of future work. It is also of interest to theoretically bound the gap between the NUM solution and the results of proposed  VQ based scheduling scheme.

\appendices
\section{Proof of Theorem \ref{theo:feasible}}\label{pf:theo-feasible}
We adopt similar techniques in the proof of Theorem 1 in \cite{BetBur14a}. 
Once the statement for one cluster size in one band is proven, the conclusion can be easily extended to general ATSs with various cluster sizes and multiple bands. Thus, we focus on clusters of size $L$ in Band-$A$ (i.e., subband $L$ in Band-$A$). We ignore the index $A$ for simplicity. All the clusters considered below satisfy $\Cc\!\subset\!\Bc^{(A)}$ and $|\Cc|\!=\!L$.

The set of feasible scheduling instants is denoted by $\mathcal{F}$, which includes vectors $\mathbf{e}$ with element  $e_{k\Cc}\!\in\!\{0,1\}$, where $e_{k\Cc}\!=\!1$ if user $k$ connects to $\Cc$ and $e_{k\Cc}\!=\!0$ otherwise. According to Defn.~\ref{feas-sched-one}, $\mathbf{e}$ is consisted of $\{e_{k\Cc}\}$ satisfying that user $k$ connects to at most one cluster and BS $j$ serves at most $S_j(\Csize)$ distinct users.
By time sharing among the feasible scheduling instants in $\mathcal{F}$, any fractional association in the convex hull of $\mathcal{F}$ can be achieved in the long term. We denote the convex hull of $\mathcal{F}$ by $X'=\conv(\mathcal{F})$ and the set of activity fractions associated to clusters in $A$ satisfying constraints in (\ref{eq:opt-cvx}) by $X$, i.e.,
\begin{equation*}
\begin{aligned}
X=&\left\{x_{k\Cc}: \sum_{\Cc: j\in \Cc}\sum_{k\in\mathcal{U}} \frac{x_{k\Cc}}{S_{j}(\Csize)} \leq 1, \sum_{ \Cc } x_{k\Cc} \leq 1,  x_{k\Cc} \geq 0,\forall k\in\mathcal{U},\forall j\in\mathcal{B}^{(A)} \text{ and }  \forall \Cc\subset\Bc^{(A)} \right\}.
\end{aligned}
\end{equation*}
It is easy to show that any feasible scheduling instants in $\mathcal{F}$ satisfies the constraints (\ref{eq:opt-ct-cluster-cvx})-(\ref{eq:opt-ct-positive-cvx}), and thus $\mathcal{F}\subseteq X$. Note that $X$ is convex. Thus, we have $X'=\conv(\mathcal{F})\subseteq X$.

As for the opposite direction (i.e., $X \not\subseteq X'$),  we first define the \emph{totally unimodular} (TU) matrix: every  square submatrix of a TU matrix has determinant $+1,-1$ or $0$. The Hoffman \& Kruskal's (1956) Theorem claims that a matrix $\mathbf{B}$ is TU if and only if for each integral vector $\mathbf{b}$, the \emph{extreme points}  of the polyhedron $\{\mathbf{z}: \mathbf{B}\mathbf{z}\!\leq\!\mathbf{b}, \mathbf{z}\!\geq\!0 \}$ are integral \cite{Sch98}. 
Denoting by $N_{\Cc L}^{(A)}$ the number of size-$L$ clusters in Band-$A$, we let $\mathbf{x}\!=\![\mathbf{x}_{1}^T, \mathbf{x}_{2}^T,\cdots, \mathbf{x}_{K}^T]^T$ with 
$\mathbf{x}_k\!=\![x_{k\Cc_1}, x_{k\Cc_1}, \cdots, x_{k\Cc_{N_{\Cc L}^{(A)}}}]^T$, and 
$\mathbf{b}\!=\!\left[S_1(L), \cdots, S_{J}(L), 1, \cdots, 1 \right]^T$ with size $(J\!+\! K)\!\times\! 1$.  We let $\mathbf{B}\!=\!\left[\begin{smallmatrix} \mathbf{C}\\ \mathbf{D}\end{smallmatrix}\right]$, where the size of matrices $\mathbf{C}$ and $\mathbf{D}$ are  $J\!\times\!(K N_{\Cc L}^{(A)})$ and $K\!\times\!(KN_{\Cc L}^{(A)})$, respectively. The element in $j$th row and $\left((k\!-\!1)N_{\Cc L}^{(A)}\!+\!i\right)$th column of matrix $\mathbf{B}$ is 1 $\forall k\!\in\!\Uc$ if $j\!\in\!\Cc_i$, and 0 otherwise. 
The matrix $\mathbf{C}$ has all elements being 1.
Recall that we consider large networks including the following type of cluster combination:  $\{j_1,j_2\}$, $\{j_1,j_3\}$ and $\{j_2,j_3\}$ if $L_{\max}^{(A)}\!>\!1$, where $j_1,j_2$ and $j_3$ are BS indexes. Then, $\mathbf{B}$ with $L_{\max}^{(A)}\!>\!1$ always includes the submatrix $\left[\begin{smallmatrix}1 & 1 & 0 \\ 1 & 0 & 1 \\ 0 & 1 & 1 \end{smallmatrix}\right]$ whose determinant is -2, and thus $\mathbf{B}$ is not TU. According to the Hoffman \& Kruskal's (1956) Theorem, there are some non-integer extreme points $\mathbf{v}\!\in\! X$ that cannot be characterized by a convex combination of any elements in $\mathcal{F}$. Thus, we have $\mathbf{v}\!\not\in\! \conv(\mathcal{F})\!=\!X'$ and~$X\!\not\subseteq\! X'$.

\section{Proof of Proposition \ref{prop:uniqueass}}\label{pf:prop-uniqueass}
We use the techniques similar to the proof of Prop.~3 in \cite{YeAnd13ABS}, where a graph is used to represent the association and KKT conditions (\ref{eq:kkt-Rk}) are used to restrict the graph structure. For a given cluster size $L$ in Band-$A$, we denote the graph by $G_1$, where nodes represent users, and  edge represents the BS cluster that serves the two nodes (users). Each node has an ID indicating the user index, while each edge has a color that identifies the BS cluster.  

Recalling that constraints (\ref{eq:opt-ct-ue-cvx}) are inactive, we have $\theta_{kL}^{(A)}\!=\!0, \forall k\!\in\!\Uc$. 
If there are two users $k$ and $m$ being served by size-$L$ clusters $\Cc_1$ and $\Cc_2$ in Band-$A$ (i.e., $x_{k\Cc_1}^{(A)}\!>\! 0$, $x_{k\Cc_2}^{(A)}\!>\! 0$, $x_{m\Cc_1}^{(A)}\!>\! 0$, $x_{m\Cc_2}^{(A)}\!>\! 0$),  we have  $R_k\!=\!\frac{r_{k\Cc_1}^{(A)}}{\sum_{j:j\in\Cc_1}\nu_{jL}^{(A)}/S_j(L)}\!=\!\frac{r_{k\Cc_2}^{(A)}}{ \sum_{j: j\in \Cc_2}\nu_{jL}^{(A)}/S_{j}(\Csize)}$ and $R_m\!=\!\frac{r_{m\Cc_1}^{(A)}}{\sum_{j:j\in\Cc_1}\nu_{jL}^{(A)}/S_j(L)}\!=\!\frac{r_{m\Cc_2}^{(A)}}{\sum_{j:j\in\Cc_2}\nu_{jL}^{(A)}/S_j(L)}$ from KKT condition (\ref{eq:kkt-Rk}), where $R_k\!=\!\sum_{A'=1}^3  \sum_{\mathcal{C'}\subset\Bc(A')} x_{k\mathcal{C'}}^{(A')} r_{k\mathcal{C'}}^{(A')}$. Thus,  we have
\begin{equation}\label{eq:kkt-2user}
\frac{r_{k\Cc_1}^{(A)}}{r_{k\Cc_2}^{(A)}}=\frac{r_{m\Cc_1}^{(A)}}{r_{m\Cc_2}^{(A)}},
\end{equation}
which is true with probability  $0$.  Therefore, it is almost sure that any two users can share at most one same cluster of size $L$ in Band-$A$. Similarly, we consider an example of three users $k, m, i$ and clusters $\Cc_1, \Cc_2, \Cc_3$. User $k$ is associated to $\Cc_1$ and $\Cc_2$, user $m$ is associated to $\Cc_1$ and $\Cc_3$, and user $i$ is associated to $\Cc_2$ and $\Cc_3$. 
We consider the following three cases:

\noindent 1) Clusters $\Cc_1$, $\Cc_2$ and $\Cc_3$ are different: we have
$
\frac{r_{k\Cc_1}^{(A)}}{r_{k\Cc_2}^{(A)}}=\frac{\sum_{j: j\in \Cc_1}\nu_{jL}^{(A)}/S_{j}(\Csize)}{\sum_{j: j\in \Cc_3}\nu_{jL}^{(A)}/S_{j}(\Csize)} \frac{\sum_{j\in \Cc_3}\nu_{jL}^{(A)}/S_{j}(\Csize)}{\sum_{j\in \Cc_2}\nu_{jL}^{(A)}/S_{j}(\Csize)} =\frac{r_{m\Cc_1}^{(A)}}{r_{m\Cc_3}^{(A)}}\frac{r_{i\Cc_3}^{(A)}}{r_{i\Cc_2}^{(A)}},
$
which is true with probability $0$.

\noindent 2) $\Cc_1=\Cc_2\neq \Cc_3$:  users $m$ and $i$ are served  by both $\Cc_1$ and $\Cc_3$, which is true with probability $0$ from~(\ref{eq:kkt-2user}).

\noindent 3) $\Cc_1=\Cc_2=\Cc_3$: users $k$, $m$ and  $i$ are served by the same cluster, which is possible. In this case, the graph  becomes a \emph{complete graph}.

Therefore, the graph $G_1$ with three users either contains a loop with the same color edges or no loop. We can get a similar result for graph $G_1$ with more than three users, where any subgraph formed by users served by the same BS cluster is a complete graph.
Thus, we generate a new graph, $G_2$, where node represents a cluster. Hence, $G_2$ has $N_{\Cc L}^{(A)}$ nodes. There is an edge between two nodes in $G_2$, if they have a common vertex in $G_1$ (i.e., there is at least one user served  by both these two clusters). Thus, the number of users who are served by more than one cluster is limited by the edge of $G_2$. Any loop in $G_2$ corresponds to a loop with more than one edge color in $G_1$. Since there are no such colorful loops in $G_1$, there is no loop in $G_2$. In other words, $G_2$ is a tree. Thus, the  maximal number of edges in $G_2$ (i.e., the maximal number of fractional users) is one less than the number of nodes (i.e., $N_{\Cc L}^{(A)}\!-\!1$). 

\section {Implementation issues of the dual-subgradient algorithm}
\label{app:complexity}
We let $L_{\max}\!=\!\max_A L_{\max}^{(A)}$, $N_{\Cc m}\!=\!\max_{L,A} N_{\Cc L}^{(A)}$ and $N_A$ be the number of operations. To solve (\ref{eq:opt-cvx}) directly by CVX \cite{cvx}, we have the problem of size $O( N_{\Cc m}N_AL_{\max} K)$, which is dominated by the size of variables  $x_{k\Cc}^{(A)}$. On the other hand, as discussed below, the proposed algorithm has lower complexity. 
Let  $L_{a}$   be the maximal number of active cluster sizes over all bands (i.e., $\max_A |\{L: \lambda_{AL}\!>\!0\}|$). The size of the LP (\ref{eq:opt-LP})  is $O(N_{\Cc m} N_A L_a\min\{N_{\Cc m}\!-\!1, K\} \!+\! N_A L_a\max\{0, K\!-\!J\!+\!1\})$, where the first term signifies the size of positive  $x_{k\Cc}^{(A)}$ for fractional users and the second term signifies the size of positive $x_{k\Cc}^{(A)}$ for users with unique association. It is easy to check that the size of (\ref{eq:opt-LP}) is smaller than the size of  (\ref{eq:opt-cvx}). 
As shown in Sec. \ref{sec:simulation},  the number of fractional users is very small (less than 3.5\%$K$), and thus the size of (\ref{eq:opt-LP})  is much smaller (less than 3.5\%) than (\ref{eq:opt-cvx}). Moreover, the size of (\ref{eq:opt-LP}) can be further reduced when $L_a/L_{\max}$ is small (e.g., only 2 active cluster sizes among 4 possible sizes in Sec. \ref{sec:simulation} ). The fast convergence in the first part (i.e., steps (\ref{eq:dual-x})-(\ref{eq:dual-nu})) of the algorithm (less than 60 iterations in our simulation) along with the low complexity per iteration, and the reduced size of  (\ref{eq:opt-LP}) makes that the proposed~algorithm can be more efficient than CVX for larger networks.

\section{Detailed Dual Subgradient Algorithm with Redundanct Contraints}\label{pf:dual-algo}
In the formulated problem (\ref{eq:opt-cvx}) in our paper, constraints (\ref{eq:opt-ct-ue-cvx})-(\ref{eq:opt-ct-mu-cvx}) imply $x_{KC}^{(A)}\leq 1$. Thus, adding the additional constraint $x_{KC}^{(A)}\leq 1$ to (\ref{eq:opt-cvx}) will not change the problem. In other words, the following problem with  constraint $x_{KC}^{(A)}\leq 1$ is equivalent to our original optimization problem (\ref{eq:opt-cvx}).

\begin{subequations}\label{eq:opt-cvx-rd}
\begin{align}
\max\limits_{\lambda_{A\Csize}, x_{k\Cc}^{(A)},\mu_A} \ & \sum_{k\in\mathcal{U}} U\left( \sum_{A=1}^{3} \sum_{\substack{
\Cc: \Cc\subset \Bc^{(A)},  \\ |\Cc|\leq L_{\max}^{(A)}}} x_{k\Cc}^{(A)}r_{k\Cc}^{(A)}\right)\\
\text{s.t. } &\sum_{\substack{\Cc: \Cc\subset \Bc^{(A)}, \\ j\in \Cc, |\Cc|=\Csize}}\frac{\sum_{k\in\mathcal{U}} x_{k\Cc}^{(A)}}{S_j(\Csize)} \leq \lambda_{AL}, \ \forall j\in \Bc^{(A)}, \forall L \leq L_{\max}^{(A)}, \forall A, \label{eq:opt-ct-cluster-cvx-rd}\\
& \sum_{\Cc: |\Cc| = L, \Cc\subset \Bc^{(A)}} x_{k\Cc}^{(A)} \leq \lambda_{AL}, \ \forall k \in \Uc, \forall L \leq L_{\max}^{(A)}, \forall A, \label{eq:opt-ct-ue-cvx-rd}\\
& x_{k\Cc}^{(A)}\in[0, 1], \ \forall k \in \Uc, \forall \Cc: |\Cc|\leq L_{\max}^{(A)}, \forall A, \label{eq:opt-ct-x-cvx2}\\
&\sum_{L=1}^{L_{\max}^{(A)}} \lambda_{AL} \leq \mu_A, \forall A,\label{eq:opt-ct-sumArch-cvx-rd}\\
&\sum_{A=1}^3 \mu_A \leq 1,\label{eq:opt-ct-mu-cvx-rd}\\
& \lambda_{AL}, \mu_A \geq 0, \ \forall k \in \Uc, \forall \Cc, \forall L\leq L_{\max}^{(A)}, \forall A,\label{eq:opt-ct-positive-cvx-rd}
\end{align}
\end{subequations}
Note that the above problem formulation is the same as (\ref{eq:opt-cvx}) in the paper, except the redundant constraint (\ref{eq:opt-ct-x-cvx2}). 
The dual subgradient algorithm in the paper is proposed based on the above equivalent optimization problem (\ref{eq:opt-cvx-rd}). \\
Specifically,  we let $\nu_{jL}^{(A)}$ and $\theta_{kL}^{(A)}$ be the Lagrange multipliers corresponding  to (\ref{eq:opt-ct-cluster-cvx-rd}) and (\ref{eq:opt-ct-ue-cvx-rd}), respectively. The dual problem of (\ref{eq:opt-cvx-rd}) is
\[\min\limits_{\nu_{jL}^{(A)}, \theta_{kL}^{(A)} \geq 0} \ \sum_{k\in\mathcal{U}} f_k\left(\nu_{jL}^{(A)}, \theta_{kL}^{(A)}\right) + g\left(\nu_{jL}^{(A)}, \theta_{kL}^{(A)}\right),\]
where 
\begin{equation}\label{eq:opt-dual-UE-rd}
\begin{aligned}
f_k\left(\nu_{jL}^{(A)}, \theta_{kL}^{(A)}\right) =\max\limits_{x_{k\Cc}^{(A)}\in[0,1]} \  &\log\left(\sum_{A=1}^3 \sum_{\substack{
\Cc: \Cc\subset \Bc^{(A)},\\ |C|\leq L_{\max}^{(A)}}}x_{k\Cc}^{(A)}r_{k\Cc}^{(A)} \right) - \sum_{A=1}^3\sum_{L=1}^{L_{\max}^{(A)}}\sum_{\substack{\Cc:\Cc\subset\Bc^{(A)},\\ |\Cc|=L}}\sum_{j: j\in\Cc} \frac{\nu_{jL}^{(A)}}{S_j(L)}x_{k\Cc}^{(A)}\\
& - \sum_{A=1}^3\sum_{L=1}^{L_{\max}^{(A)}}\theta_{kL}^{(A)} \sum_{\Cc:\Cc\subset\Bc^{(A)}, |\Cc|=L} x_{k\Cc}^{(A)},
\end{aligned}
\end{equation}
and 
\begin{equation}\label{eq:opt-dual-lambda-rd}
g(\nu_{jL}^{(A)}, \theta_{kL}^{(A)}) =
\max\limits_{\substack{\sum_{L=1}^{L_{\max}^{(A)}}\lambda_{AL}\leq \mu_A,\\ \sum_{A=1}^3 \mu_A\leq 1} } \sum_{A=1}^3\sum_{L=1}^{L_{\max}^{(A)}}\left( \sum_{j: j\in \Bc^{(A)}} \nu_{jL}^{(A)}  + \sum_{k\in\Uc}\theta_{kL}^{(A)}\right)\lambda_{AL}.
\end{equation}
The function (\ref{eq:opt-dual-UE-rd}) is simmilar to (\ref{eq:opt-dual-UE}) in the paper, except that we have an additional contraint $x_{KC}^{(A)}\leq 1$. 

The constraints of (\ref{eq:opt-cvx-rd}) satisfy the Slater condition, and thus strong duality holds (i.e., the dual problem  and the original problem (\ref{eq:opt-cvx-rd}) have the same optimal value, which has the same optional solution with problem (\ref{eq:opt-cvx}) in the paper).

The optimization problem  (\ref{eq:opt-dual-UE-rd})  has the closed-form optimal solution
\begin{equation}\label{eq:dual-x-rd}
x_{k\Cc}^{(A)} =
\begin{cases}
\left[\frac{1}{\sum_{L:L=|\Cc|}\left(\sum_{j:j\in\Cc}\nu_{jL}^{(A)}/S_j(L)+\theta_{kL}^{(A)}\right)}\right]^1_0, & \text{if } \{\Cc,A\} = \{\Cc^*, A^*\},\\
0, & \text{otherwise},
\end{cases}
\end{equation}
where $[x]_0^1= \min\{1, \max\{0, x\}\}$, $\{\Cc^*, A^*\}=\arg\max_{\Cc, A}  \frac{r_{k\Cc}^{(A)}}{\sum_{L:L=|\Cc|}\left(\sum_{j:j\in\Cc}\nu_{jL}^{(A)}/S_j(L)+\theta_{kL}^{(A)}\right)}$\footnote{If we have multiple pairs of $\{\Cc^*,A^*\}$, we can pick the pair with largest $r_{k\Cc}^{(A)}$.}. 

The problem (\ref{eq:opt-dual-lambda-rd}) is an LP and one optimal solution is\footnote{If we have multiple $\{A,L\}$ pairs that maximize the $\sum_{j: j\in \Bc(A)} \nu_{jL}^{(A)}  + \sum_{k\in\Uc}\theta_{kL}^{(A)}$, we just randomly pick one.}
\begin{equation}\label{eq:lambdaopt-dualalgo-rd}
\lambda_{AL} = \left\{
\begin{aligned}
& 1, \text{ if } \{A,L\}= \arg\max_{A',L'} \sum_{j: j\in \Bc(A')} \nu_{jL'}^{(A')}  + \sum_{k\in\Uc}\theta_{kL'}^{(A')}, \\
&0, \text{ otherwise},
\end{aligned}\right.
\end{equation}
and 
\begin{equation}\label{eq:muopt-dualalgo-rd}
\mu_A = \left\{
\begin{aligned}
& 1, \text{ if there exists a band $A$ such that the above } \lambda_{AL}>0, \\
&0, \text{ otherwise}.
\end{aligned}\right.
\end{equation}

The $t$th iteration of the  algorithm is as follows.
\begin{enumerate}
\item Update the activity fractions by (\ref{eq:dual-x-rd}).
\item  Update resource allocation for different bands and clusters by (\ref{eq:lambdaopt-dualalgo-rd}) and (\ref{eq:muopt-dualalgo-rd}). 
\item Update the Lagrangian multipliers by 
\begin{equation}\label{eq:dual-nu-rd}
\nu_{jL}^{(A)}(n+1) = \left[\nu_{jL}^{(A)}(n) - \delta(n) \left(\lambda_{AL}(n) - \sum_{\substack{\Cc: \Cc\subset \Bc(A),\\ j\in \Cc, |\Cc|=L}}\frac{\sum_{k\in\mathcal{U}} x_{k\Cc}^{(A)}(n) }{S_{j}(\Csize)}\right)\right]^+,
\end{equation}
and 
\begin{equation}\label{eq:dual-theta-rd}
\theta_{kL}^{(A)}(n+1) = \theta_{kL}^{(A)}(n) - \delta(n) \left(\lambda_{AL}(n) - \sum_{\Cc: \Cc\subset\Bc(A), |\Cc|=L} x_{k\Cc}^{(A)} \right),
\end{equation}
where  $[z]^+=\max\{z, 0\}$ and $\delta(n)$ is the stepsize at the $n^{\rm th}$ iteration.
\end{enumerate}

From the above steps, we can observe that the difference in the dual algorithm by adding the redundant constraint $x_{KC}^{(A)}\leq 1$  is in (\ref{eq:dual-x-rd}), where the variable $x_{k\Cc}^{(A)}$ needs to be projected to the set [0, 1]. Based on this constraint, the subgradients $\Delta \nu_{jL}^{(A)}$ and $\Delta \theta_{kL}^{(A)}$ are bounded, which can be used to show the convergence of the dual algorithm.


\bibliographystyle{ieeetr}
\bibliography{qyeallbib}

\end{document}